\documentclass[10pt,twocolumn,twoside]{IEEEtran}
\IEEEoverridecommandlockouts

\usepackage{cite}
\usepackage{amsmath,amssymb,amsfonts,mathtools}
\usepackage{amsthm}
\usepackage{algorithmic}
\usepackage{graphicx}
\usepackage{textcomp}
\usepackage{xcolor}
\usepackage[linesnumbered,ruled,vlined]{algorithm2e}
\usepackage[shortlabels]{enumitem}
\usepackage{bm}
\usepackage{stfloats}
\usepackage{dsfont}

\def\BibTeX{{\rm B\kern-.05em{\sc i\kern-.025em b}\kern-.08em
    T\kern-.1667em\lower.7ex\hbox{E}\kern-.125emX}}
\newtheorem{theorem} {Theorem}
\newtheorem{example} {Example}
\newtheorem{definition} {Definition}
\newtheorem{assumption} {Assumption}
\newtheorem{lemma} {Lemma}
\newtheorem{coro} {Corollary}
\newtheorem{proposition}{Proposition}
 \newtheorem*{remark}{Remark}

\newcounter{subeqn} %

\usepackage{tikz}
\usepackage{pgfplots}
\pgfplotsset{every axis/.append style={line width=1pt}}
\pgfplotsset{every tick label/.append style={font=\tiny}}
\pgfplotsset{compat=newest}
\usetikzlibrary{shapes, fit, decorations.markings, matrix, plotmarks, positioning, spy, patterns, shadows, calc, backgrounds, arrows.meta, arrows}
\usetikzlibrary{decorations.pathreplacing}
\usetikzlibrary{circuits.ee.IEC}
\usetikzlibrary{fit}
\usetikzlibrary{plotmarks,pgfplots.colorbrewer}
\usetikzlibrary{decorations.pathreplacing}
\usetikzlibrary{circuits.ee.IEC}
\usepackage[americanvoltages,siunitx]{circuitikz}

\newcommand{\ts}[1]{{\textnormal{#1}}}

\newcommand{\ie}{\emph{i.e.},\ }

\newcommand{\Nset}{\mathbb{N}}
\newcommand{\Rset}{\mathbb{R}}

\newcommand{\mc}{\mathcal}
\newcommand{\md}{\mathds}
\newcommand{\mb}{\mathbf}
\newcommand{\mbb}{\mathbb}
\newcommand{\mf}{\mathfrak}

\DeclareFontFamily{U}{mathx}{\hyphenchar\font45}
\DeclareFontShape{U}{mathx}{m}{n}{
      <5> <6> <7> <8> <9> <10>
      <10.95> <12> <14.4> <17.28> <20.74> <24.88>
      mathx10
      }{}
\DeclareSymbolFont{mathx}{U}{mathx}{m}{n}
\DeclareFontSubstitution{U}{mathx}{m}{n}
\DeclareMathSymbol{\bigtimes}{1}{mathx}{"91}

\usepackage[nolist]{acronym}
\begin{acronym}
  \acro{uG}[$\mu$G]{Micro-Grid}
  \acro{DSM}{Demand-Side Management}
  \acro{P2P}{Peer-to-Peer}
\end{acronym}

\allowdisplaybreaks
\begin{document}

\title{A Coalitional Game for Demand-Side Management in a Micro-Grid with Multiple Electricity Retailers
\thanks{P. R. Baldivieso-Monasterios is with the School of Electric and Electronic  Engineering, The University of Sheffield, Mappin Street, S1 3JD Sheffield, United Kingdom. {\tt\small p.baldivieso@sheffield.ac.uk}}
\thanks{G Konstantopoulos is with the Department of Electrical and Computer Engineering, University of Patras, Rion, 26500, Greece. 
{\tt\small g.konstantopoulos@ece.upatras.gr}}%
\thanks{D. Bauso is with the Dipartimento di Ingegneria, Universit\`a di Palermo, viale delle Scienze, Palermo, Italy {\tt\small dario.bauso@unipa.it}}%
}

\author{Pablo R. Baldivieso-Monasterios*, Fernando Genis Mendoza, George Konstantopoulos, and Dario Bauso}

\maketitle

\begin{abstract}
  This paper develops a demand-side management framework for electricity networks with multiple competing retailers. The interaction among retailers is formulated as a coalitional game, yielding a family of coupled mixed-integer optimisation problems in which retail prices, consumer power demands, and the network partition are jointly optimised. To solve this problem, we propose a coalition-formation algorithm based on multi-objective optimisation principles. The algorithm seeks to identify coalition structures that balance retailer profit and consumer welfare. We prove that the proposed algorithm converges in a finite number of steps and recovers a subset of weakly Pareto-efficient solutions of the coupled optimisation problems. The framework is further extended to a risk-sharing formulation, in which the objective is defined using conditional value-at-risk. Numerical simulations on an academic example demonstrate the method's behaviour and show that the resulting equilibrium partition set contains several admissible trade-offs between the competing objectives. The results provide a tractable approach for analysing competition, coalition formation, and risk-aware pricing in multi-retailer demand-side management systems.
\end{abstract}


\section{Introduction}
The problem of \ac{DSM} has become more relevant in recent years, driven by growing calls for greater flexibility in power demand at the distribution level, driven by large numbers of flexible loads, distributed generation units, and local storage devices. Consumers can influence how a power network operates by optimising their demand in response to price incentives and contractual arrangements. Similarly, energy retailers seek to maximise their gains from their energy production whilst attracting the largest number of users. In \ac{DSM}, this retailer-buyer interaction and coordination enable the network to respond to and adjust its performance in the face of unforeseen events.  

The landscape of dynamic pricing algorithms in the literature is vast. In this brief literature review, we describe pricing mechanisms in game-theoretic terms for \ac{DSM}. The consumer-retailer interaction in \ac{DSM} can be modelled using game-theoretic methods. The classic approaches include the seminal work of  \cite{Mohsenian-Rad2010}, in which network participants employ their best-response maps to prices. The outcome is a Nash equilibrium, and the resulting game links individual cost minimisation with aggregate load shaping. The authors in \cite{Samadi2012} complemented \cite{Mohsenian-Rad2010} by introducing notions of social welfare in \ac{DSM}. One caveat mentioned in the above papers but not addressed is the hierarchical nature of \ac{DSM}, in which an energy provider sets prices and consumers respond optimally. Stackelberg formulations, including \cite{Maharjan2013b}, exploit this hierarchical structure. The ideas of \ac{DSM} have been incorporated in a control framework in \cite{Stephens2015} to coordinate distributed assets. These classical approaches, and their references therein, establish prices as the natural coordination signal in a \ac{DSM}; however, they implicitly assume that the structure between retailers and consumers is fixed. 

Work on \ac{DSM} has shifted towards the \ac{uG} case, \ie smaller, more localised networks dominated by uncertain renewable energy sources. The price-based coordination strategy in \cite{Quijano2023} uses prices to coordinate electric springs for \ac{DSM} in \ac{uG}s, showing how device-level flexibility can be activated through an economic signal. As the ideas for \ac{DSM} mature, research, such as \cite{Yuan2023a} and \cite{Qiu2024}, focuses on mitigating uncertainties introduced by renewable generation and demand. These schemes deliver robust, risk-aware pricing policies for managing energy flows. These contributions complement classical \ac{DSM} by introducing uncertainty and operational risk into price-based energy management. However, the market structure remains largely fixed: prices and schedules are optimised, but the assignment of buyers to competing retailers is not itself a decision variable.

Within this \ac{uG} setting, \ac{DSM} has met \ac{P2P} methods for handling and coordinating prices and load consumption. In the \ac{P2P} framework, see \cite{Tushar2023}, the concept of prosumer, a node that acts as both a generator and a load, takes centre stage. Under \ac{P2P}, prosumers exchange resources subject to network constraints; the formulation and modelling are similar to those in \ac{DSM}: a game-theoretic framework. For example, \cite{Belgioioso2022} proposes energy trading as a game and establishes convergence guarantees to an efficient, operationally feasible equilibrium. The framework in \cite{Tarashandeh2024} also considers \ac{P2P} trading under distribution-network constraints while preserving the independent nature of agents. More recent prosumer-centric formulations, such as \cite{Hoque2024}, incorporate network-secure export--import limits into \ac{P2P} trading, whereas \cite{Xia2023} designs a grid-friendly pricing mechanism to support the diffusion of P2P energy sharing in communities. Similarly, \cite{BaldiviesoMonasterios2022} and \cite{Toderean2023} show how to coordinate pricing policies in the presence of exogenous signals such as forecasts. Despite the success of \ac{P2P}, its focus is mainly on trade clearing, operating envelopes, or bilateral exchange among prosumers, rather than on multi-retailer \ac{DSM} with coalition-dependent price and demand decisions.

Coalitional approaches provide another perspective on local energy coordination, as they enable the formation of groups of agents that jointly optimise their objectives. In \cite{Raja2023}, bilateral P2P trading is formulated as an assignment and coalitional game, yielding stable and fair contracts between buyers and sellers via distributed negotiation. The coalition-formation model of \cite{Zhang2024} studies how \ac{P2P} trading coalitions emerge in transactive energy communities and analyses the incentives for prosumer participation. Cooperative-game approaches to energy communities, such as \cite{Bossu2024}, focus on the generation and distribution of gains produced by collective self-consumption. In \cite{GenisMendoza2021}, the authors use a Stackelberg method within a coalitional framework. In \cite{Xu2014} and \cite{Ghiasvand2022}, we can see, from a numerical perspective, the relationship between coalition formation and changes in energy prices. A study of the case in which greedy prosumers do not align with the \ac{uG}'s decision is presented in \cite{Han2019}; here, a balanced game is proposed without the need to compute the imputation set. A notion of fairness is introduced by \cite{Moafi2023} using particle swarm optimisation and using the nucleolus as a game solution concept. A bidding system for cooperating prosumers is presented in \cite{Chakraborty2019}; constraints on power capacity and losses are ignored. The use of evolutionary game theory in conjunction with coalitions is proposed in \cite{Mondal2017}, where the price is a quadratic function of the consumption. These works complement price-based \ac{DSM} by explaining how cooperation can form and how benefits can be allocated. Nevertheless, they do not consider the retailer-indexed partition of buyers as the central object \ac{DSM} optimisation.

We have, however, identified the absence of a framework that jointly models price management, buyer--retailer coalition formation, and multi-objective \ac{DSM}. Recent work on \ac{DSM} has developed price-based coordination, risk-aware aggregation, and robust \ac{uG} energy management, while the coalitional energy-trading literature has focused mainly on \ac{P2P} contracts, prosumer coalitions, and fair value allocation. Existing methods do not treat the assignment of buyers to competing retailers as a decision variable, alongside continuous prices and demand. The partitioning of network conditions determines how prices are determined within each group and how the resulting coalition structure affects the objectives of all participants.

This paper tackles this gap. We treat the coalition structure itself as a decision variable in a multi-retailer \ac{DSM} problem. We pose a mixed-integer multi-objective optimisation problem that captures the notions of competition and collaboration, and solve it algorithmically. The resulting algorithm generates an equilibrium set of terminal partitions, and the scalarised objective values associated with this set provide weakly Pareto-optimal support for the competing buyer and retailer objectives. Our approach aims to provide a rigorous framework for modern energy trading platforms, such as those in \cite{GoodEnergyLTD2020} and \cite{UswitchLimited}. The contributions and structure of this paper are as follows.

\begin{itemize}
\item We propose a mixed discrete-continuous \ac{DSM} problem in a coalitional game framework in Section~\ref{sec:analys-solut-ocp}, where multiple retailers compete in a \ac{uG} to attract consumers. We introduce a novel method that links a cost-saving game directly to the network conductance matrix.  
\item We propose an algorithm for coalition formation in Section~\ref{p4sec:game} that exhibits finite convergence, thereby maximising profits and welfare for both retailers and consumers. The final partition is not only optimal in a Pareto sense but also stable in the game-theoretic sense.
\item In Section~\ref{p4sec:game}, we characterise the resulting equilibrium partition set and relate its objective vectors to supported weak Pareto optima of the multi-objective \ac{DSM} problem
\item In Section~\ref{p4sec:stats}, we extend the formulation to stochastic demand by introducing coalition-level overload risk and a CVaR-type risk functional, thereby linking coalition formation with tail-risk reduction in retailer-specific demand portfolios.
\end{itemize}
\textbf{\emph{Notation}}: For any finite set $\mc{B}$, $|\mc{B}|$ denotes its cardinality. A \emph{game} is a triplet $(\mc{N},\{\mc{H}_i\}_{i\in\mc{N}},\{u_i\}_{i\in\mc{N}})$ composed by a set of players $\mc{N}$, a collection action sets $\mc{H}_i$, and a collection of payoff functions $u_i\colon\mc{H}_i\to\Rset$. An \emph{action profile} for a set of $|\mc{N}|$ players is an ordered tuple $(h_1,\ldots,h_{|\mc{N}|})$. The \emph{best response set} of player $i\in\mc{N}$ is defined as $\mc{Q}_i(h_{-i}):=\arg\max_{h_i \in \mathcal{H}_i}u_i(h_i,h_{-i})$ where $h_{-i} = (h_j)_{j\in\mc{N}\setminus\{i\}}$. An  \emph{action profile} $(h_1^*,\ldots,h_{|\mc{N}|}^*)$ is a \emph{Nash Equilibrium} if $h_i^*\in\mc{Q}_i(h_{-i}^*)$ for all $i\in\mc{N}$. A coalitional game is defined by the pair $(\mathcal{N},\nu)$ with $\mc{N}$ a set of players and $\nu\colon 2^{\mathcal{N}}\to\Rset$ its characteristic function. A payoff vector is an element $a\in\Rset^{|\mc{N}|}$ such that $a_i\in\Rset$ denotes the gain of player $i\in\mc{N}$. The \emph{core}, $\mc{C}\subset\Rset^{|\mc{N}|}$, of a coalitional game $(\mc{N},\nu)$ satisfies $\mc{C} = \{a\in\Rset^{|\mc{N}|}\colon\forall S\subset\mc{N},~\sum_{i\in S}a_i\geq\nu(S)\}$ \cite{Bausoa}. A cooperative game $(\mc{N},\nu)$ with $\nu(\emptyset)= 0$ is: a \emph{convex game} if $\nu(S\cup T)+\nu(S\cap T)\leq \nu(S)+\nu(T)$ $\forall S,T\subseteq \mc{N}$; a \emph{permutationally convex} (PC)  game if there exists a permutation $\pi$ of $\mc N$ such that $\forall S\subseteq \mc{N}\setminus [k]$ and $k\geq j$, $\nu([j]\cup S)-\nu([j])\leq \nu([k]\cup S)-\nu([k])$ where $[k] = \{k\}\cup\{i\in\mc{N}\colon \pi(i)<\pi(k)\}$. A power network is defined by a connected, undirected, and weighted graph $\mc{G} = \mc{(N,E)}$, where $\mc{N}$ is the set of nodes, and $\mc{E}\subseteq\mc{N}\times\mc{N}$ the set of edges defining the interconnection topology. A $k-$path from node $i$ to $j$ is a sequence of $k$ edges $\phi(i,j) = \{e_1,\ldots,e_k\}\subset\mc{E}$. The mapping $\omega\colon\mc{E}\to\Rset$ defines the weight of each edge such that  $\omega(i,j) = \omega_{ij}\in\Rset$. For undirected graphs, the mapping $\omega(\cdot,\cdot)$ is symmetric and can be characterised by a symmetric adjacency matrix $A\in\Rset^{|\mc{N}|\times|\mc{N}|}$. The matrix $B\in\Rset^{|\mc{N}|\times|\mc{N}|}$, such that $B_{ij} = 1$ if $A_{ij}>0$ and $B_{ij} = 0$ otherwise, and $A$ determine the connectivity properties of $\mc{G}$. From \cite{bullo}, $B^k$ and $A^k$, with $k\in\Nset$, denote the number of $k$-paths and the sum of the products of the weights of all $k$-paths, respectively, from node $i$ to node $j$. The total weight of $\mathcal{G}$ is $\omega(\mathcal{E})=\sum_{(i,j)\in\mathcal{E}} \omega(i,j) = \frac{1}{2}\sum_i\sum_jA_{ij}$. The out degree and Laplacian matrices of $\mc{G}$ are $D\coloneqq \ts{diag}(\sum_{j=1}^N A_{ij})$, and $\mc{L} = D-A$ respectively. The set of neighbours of node $i$ is $\mc{N}_i = \{j\in\mc{N}\colon\exists (i,j)\in\mc{E}\}$; the out-degree of node $i$ is $\delta_i = |\mc{N}_i|$. A tree is an undirected graph in which a sequence of edges connects any two vertices. The \emph{Minimum Spanning Tree} (MST) of a graph $\mathcal{G(N,E)}$ is a tree $T=(\mathcal{N, E^*})$ such that ${\mc{E}^*} = {\arg\min}\{\omega(\mc{E}')\colon \mc{E}'\subset\mc{E},\forall i\in\mc{N},\exists j\in\mc{N}, (i,j)\in\mc{E}'\}$.
\section{System Model, Definition and Preliminaries} \label{p4sec:models}
In this section, we make precise the notions of retailers and consumers. In our demand-side management problem, electricity prices regulate user consumption, with retailers and a subset of buyers jointly optimising price and consumption. 
\subsection{Coalitional setting} \label{p4sec:setscoalitions}
The set of players is $\mc{N} = \{1,\ldots,N\}$. The set $\mc{N}$ is partitioned into two non-overlapping sets: retailers $\mathcal{R}\subset\mathcal{N}$ and consumers $\mathcal{B}\subset\mathcal{N}$. Both $\mc{R}$ and $B$ satisfy  $\mc{R}\cup\mc{B}=\mc{N}$, $\mc{R}\cap\mc{B}=\emptyset$, and $N = |\mc{B}| + |\mc{R}|$. Each retailer $r\in\mc{R}$ seeks to attract a subset of consumers $\mc{B}_r\subset\mc{B}$; the \emph{retailer coalition} for $r\in\mc{R}$ is $S_r = \{r\}\cup\mc{B}_r\subset\mc{N}$. We invoke the following regularity assumption:
\begin{assumption}
  The family of sets $\{S_r\}_{r\in\mc{R}}$ satisfy the following assertions.
  \begin{enumerate}
  \item\label{p4ass:union} $\bigcup_r S_r=\mathcal{N}$,
  \item\label{p4eq:nooverlap} For any $r,s\in\mc{R}$, $S_r\cap S_s = \emptyset$,    
  \item\label{p4eq:oneseller} For any $r\in\mc{R}$, $S_r\cap \mc{R}\setminus\{r\} = \emptyset$.
  \end{enumerate}
  \label{assump:retailer_coalition}
\end{assumption}
Assumption \ref{assump:retailer_coalition} ensures the family $\{S_r\}_{r\in\mc{R}}$ is a non-overlapping covering of $\mc{N}$ and each consumer $b\in\mc{B}$ is assigned to only one $S_r$. 
\subsection{Consumer and Retailer Profit Functions}
\label{p4sec:players}
In our problem setting, consumers and retailers are considered to be price-taking rational agents, \ie both aim to maximise their profits for producing or consuming energy. The profit function for retailer $r\in\mc{R}$ is $ \Pi _{r} (\lambda_r,P)= {\lambda_r} P - C_r(\lambda,P+P_r^{loss}),$  where $\lambda_r$ is the power price and $C_r(\cdot)$ is a function corresponding to the cost of producing $P+P_r^{loss}$ units of power and loss $P_r^{loss}$. Similarly, every consumer $b$ that has opted to consume from~$r$ calculates its profit in accordance to $\Pi _{b}(P,\lambda_r)= U_b(P,r) - {\lambda_r}P$ where $U_b(\cdot,\cdot)$ is the utility from consuming $P$ power at a price $\lambda_r$. Following~\cite{Namerikawa}, we require:
\begin{assumption}
   For each $r\in\mc{R}$ and $b\in\mc{B}$, $C_r(\cdot)\colon\Rset\to\Rset$ and $U_b(\cdot)\colon\Rset\to\Rset$ are continuous and  monotonically increasing. In addition, $C_r(\cdot)$ is convex and $\forall r\in\mc{R}$, $U_b(\cdot,r)$ is concave.
   \label{assum:costs}
 \end{assumption}
 Each retailer in our setting seeks to maximise profits by selling power to a subset of buyers $\mc{B}_r^*$ at a price $\lambda_r^*$. Similarly, each consumer seeks to maximise its utility by consuming ${P_{{b}}}^*$ power from a selected $r\in\mc{R}$. This process is captured in the following coupled optimisation problems: 
 \begin{subequations}
   \begin{align}
     \mbb{O}_r(B,P_{B}) &\colon \max\{\Pi _{r} (\lambda,\sum_{b\in {C}} P_{{b}})\colon  \lambda  \in [\underline \lambda  ,\bar \lambda ],~{C}\subset B,\nonumber \\
     &\quad\quad\quad\quad \biggr|\sum_{b\in C} P_{{b}}+P_r^{loss}\biggl|\leq P_r^\ts{max}  \},\label{p4eq:optprice_mi}\\
     \mbb{O}_b(\lambda_\mc{R}) &\colon \max\{\Pi_b(P,\lambda_r)\colon P  \in [0 ,\bar \zeta ],~r\in\mc{R}\}  ,\label{p4eq:optcons_mi}
   \end{align}\label{p4eq:opt_mi}
 \end{subequations}
where $B\subseteq\mc{B}$ is the set of buyers willing to purchase energy from $r$ and $P_{B} = \{P_b\}_{b\in B}$ is the collection their power demands. For equation~\eqref{p4eq:optcons_mi}, the parameters are the prices announced by each retailer, \ie $\lambda_\mc{R} = \{\lambda_r\}_{r\in\mc{R}}$.  The optimisation problems defined in \eqref{p4eq:opt_mi} are mixed-integer, making their use in an online setting problematic. This problem, however, is of interest conceptually because its solution yields both the optimal price for each retailer $\lambda_r^*$, the optimal consumption for each consumer ${P_b}^*$, and the optimal partition $\{\mc{B}_r^*\cup\{r\}\}_{r\in\mc R}$ of $\mc N$.

\section{Analysis of the solution of the Multiple Retailer optimisation}
\label{sec:analys-solut-ocp}

The domain of the multi-objective optimisation problem comprises a vector of prices for all retailers $\lambda_{\mc{R}} = (\lambda_{1},\ldots,\lambda_{|\mc{R}|})\in\Rset^{|\mc{R}|}$, a vector of power demands for consumers $P_\mc{B} = (P_{1},\ldots,P_{|\mc{B}|})\in\Rset^{|\mc{B}|}$, and a partition of the set $\mc{N}$ into coalitions $\mc{C} = \{S_r\}_{r\in\mc{R}}$ satisfying Assumption~\ref{assump:retailer_coalition}. For both prices and powers, the set containing all feasible combinations is $\mc{X}(\mc{C}) \subset \Rset^{|\mc{R}| + |\mc{B}|}$ determined by the constraint on continuous variables of Problems~\eqref{p4eq:opt_mi}. On the other hand, the set of partitions of $\mc{N}$ is denoted by $\mc{M}_\mc{N}$. The set of  partitions satisfying Assumption~\ref{assump:retailer_coalition} is denoted $\Pi_{\mc R,B}$. Therefore, the joint set of decision variables is $(x,\mc{C}) \in \Gamma\times\mc{M}_{\mc R,B}$.

The cardinality of $\mc{M}_\mc{N}$ scales with the Bell number of $|\mc{N}|$ and undergoes a combinatorial explosion as $|\mc{N}|$ increases. In our setting, however, we do not need all the possible partitions. The next result yields the exact number of partitions enforced by Assumption~\ref{assump:retailer_coalition}.
\begin{proposition}[Viable coalitions]
Suppose Assumption~\ref{assump:retailer_coalition} holds. The number of viable partitions is $|\mc{M}_{\mc R,B}| = |\mc R|^{|\mc B|}$.
\label{prop:counting}
\end{proposition}
\begin{proof}
  The set $\mc{N}$ admits a partition $\mc{N} = \mc{R}\cup\mc{B}$. The process of assigning a subset of buyers to a retailer under Assumption~\ref{assump:retailer_coalition} entails imposing constraints on $\Pi_\mc{B}$. For each element in $\mc{C}\in\Pi_{\mc{R},\mc{B}}$, the retailer assignation correspond to a surjective map $\mc{R}\to\mc{C}$. Furthermore, if $|\mc{C}| > |\mc{R}|$, a surjective map does not exist between $\mc{R}$ and $\mc{C}$. We note that, following \cite{Stanley2011}, the number of partitions of $\mc B$ with $k>0$ elements is given by the sterling numbers of the second kind, \ie \(\genfrac\{\}{0pt}{}{|\mc{B}|}{k} = \frac{1}{k!}\sum_{h = 0}^{k}(-1)^{k-h}\binom{k}{h}h^{|\mc{B}|}\). Fixing a partition $\mc{C}\in\Pi_\mc{B}$ of $k$ elements. We can select then $k$ elements from $\mc{R}$. The number of bijections between these two sets $\{r_1,\ldots,r_k\}\to\mc{C}$ is given by $k!$. Therefore, the number of possible assignments of $k$ retailers to a partition of $k$ elements is $\genfrac\{\}{0pt}{}{|\mc{B}|}{k}\binom{|\mc{R}|}{k}k!$. The number of possible partition pairings  $|\mc{M}_{\mc R,B}|$ is \[
    \begin{split}
      ~& = \sum_{k=1}^{|\mc{R}|}k!\binom{|\mc{R}|}{k}\genfrac\{\}{0pt}{}{|\mc{B}|}{k}
                      = \sum_{k=1}^{|\mc{R}|}\sum_{h = 0}^{k}\binom{|\mc{R}|}{k}\binom{k}{h}h^{|\mc{B}|}(-1)^{k-h}\\
                      & = \sum_{k=0}^{|\mc{R}|}\sum_{h = 0}^{k}\binom{|\mc{R}|}{h}\binom{|\mc{R}| - h}{k -h}h^{|\mc{B}|}(-1)^{k-h}\\
                      & = \sum_{h=0}^{|\mc{R}|}\binom{|\mc{R}|}{h}h^{|\mc{B}|}\sum_{j = 0}^{|\mc{R}|-h}\binom{|\mc{R}| - h}{j}(-1)^{j}\\
      & = \sum_{h=0}^{|\mc{R}|}\binom{|\mc{R}|}{h}h^{|\mc{B}|}(x + 1)^{|\mc{R}|-h}\vert_{x=-1} = |\mc{R}|^{|\mc{B}|}
    \end{split}
  \]
In the second equality, we replaced the definition of Stirling number; in the third, starting the sum from $k=0$ does not change the sum and used $\binom{a}{b}\binom{b}{c} = \binom{a}{c}\binom{a-c}{b-c}$; in the fourth, we relabel $ j = k-h$ and change the limits; in the fifth one, we use the binomial theorem; and in the sixth one, we note that the binomial is only nonzero when $h = |\mc{R}|$.   
\end{proof}
Furthermore, the set $\mc{M}_\mc{N}$ is endowed with a partial order given by the refinement, see \cite{Baldivieso-Monasterios2021}. Two partitions $\mc{C},\mc{D}\in\mc{M}_\mc{N}$ are said to be comparable under refinement $\mc{C}\preceq\mc{D}$ if $\forall c\in\mc{C}$ $c\subset d$ for some $d\in\mc{D}$. The pair $(\mc{M}_\mc{N},\preceq)$ represents a partial order. A subset $\mf{C}\subset\mc{M}_\mc{N}$ is a \emph{chain} if any pair $\mc{C},\mc{D}\in\mf{C}$ can be compared; the set $\mf{C}$ is an \emph{anti-chain} if none of its elements could be compared. This leads to the following result.
\begin{proposition}
  Suppose Assumption~\ref{assump:retailer_coalition} holds. The set $\mc{M}_{\mc R,B}\subset\mc{M}_\mc{N}$ is an \emph{anti-chain} for the order relation $\preceq$.
  \label{prop:anti_chain}
\end{proposition}
\begin{proof}
Suppose $\mc{C},\mc{D}\in\mc{M}_{\mc R,B}$. Each element $c\in\mc{C}$ has the structure $x = \{r\}\cup{B}_r$ with $r\in\mc{R},~B_r\subseteq\mc{B}$ and similarly for $\mc{D}$. From this, the cardinality of $\mc{C}$ and $\mc{D}$ are exactly $|\mc{R}|$. From the definition of $\preceq$, $\mc{C}\preceq\mc{D}$ iff $\forall c\in\mc{C},\exists d\in\mc{D},~\ts{s.t. }c\subseteq d$. Fixing any $c = \{r\}\cup B_r$, the only possible element in $\mc{D}$ containing it has the form $d = \{r\}\cup \tilde{B}_r$ with $B_r\subseteq \tilde{B}_r$. Now, any  $b \in \tilde{B}_r\setminus B_r$, by Assumption~\ref{assump:retailer_coalition}, is contained within a $c'\in\mc{C}\setminus \{c\}$ such that $c' = \{r'\} \cup \{b\}\cup \{B_{r'}\}$. However $b\in d = \{r\}\cup\tilde{B}_r$, which implies that $c'$ is not contained in any element of $\mc{D}$. Since $c\in\mc{C}$ was chosen arbitrarily, the claim follows. 
\end{proof}
The consequence of Proposition~\ref{prop:anti_chain} is that we can see the set $\mc{M}_{\mc R,B}$ as the set of functions $\psi\colon\mc{B}\to\mc{R}$. We can use this fact to define a method of ordering this anti-chain using a ``buyer reassignment'' idea. Consider a partition defined as a function $\psi$, the function $\psi'$ defined as $\psi'(b) = r'$ and $\psi'(a) = \psi(a)$ for all $a\neq b$. This corresponds only to a consumer changing their retailer. The partitions defining $\mc{C}\sim\psi$ and $\mc{D}\sim\psi'$ can be compared using a pre-order $\preceq_\eta$, \ie $\mc{C} \preceq_\eta \mc{D}$ iff $\eta(\mc{C}) \leq \eta(\mc{D})$ for a function $\eta\colon\mc{M}_{\mc R,B}\to\Rset$. In the next sections, we will define this scalar function that allows us to compare partitions. 
\section{Coalitional Game with Multiple Retailers} \label{p4sec:game}
In this section, we propose an algorithmic way of handling the mixed-integer coupled optimisation problems $\mbb{O}_r(\cdot)$ and $\mbb{O}_b(\cdot)$ for each $r,b\in\mc{N}$. A direct solution can be prohibitively difficult to solve as the number of buyers and retailers grows. Algorithm~\ref{alg:coal_form} defines the coalition-forming procedure that yields a Pareto-optimal solution. 
\subsection{Retailer network and coalitional game}
\label{sec:cost_network}
Each retailer $r\in\mc{R}$ seeks to attract buyers to its coalition by means of incentives, \ie by forming the \emph{retailer cost network} $\mc{G}_r = (S_r,\mc{E}_r)$. The set of edges $\mc{E}_r$ comprises links of the form $(r,b)$ and $(b,d)$ for $d,b\in\mc{B}_r$. The weights $\omega(r,b)$ determine the \emph{direct connection} costs and is the geometric average of all $n-$paths from $r$ to $b$ within $\mc{G}$ such that \[\omega(r,b)=\gamma \biggr(\frac{[A^n]_{rb}}{[B^n]_{rb}}\biggl)^{\frac{1}{n}} + n \varepsilon.\]If $(r,b)\in\mc{E}$, then $n=1$; otherwise, we increase $n$ until we connect $r$ and $b$ by a path of length $n$. The constants $\gamma > 0$ and $\varepsilon > 0$ correspond to a connection fee, and the matrices $ A$ and $ B$ are the adjacency and binary adjacency matrices of $\mc{G}$. We note that the cost increases linearly with the number of walks needed to connect two nodes. The links of the form $(b,d)$ with $b,d\in\mc{B}_r$ represent \emph{aggregate connection} and may represent benefits of joining a coalition. The weight is \[\omega(b,d) = \gamma A_{{b}{d}} + \beta \varepsilon,\] if $A_{{b_i}{b_j}}>0$ where $\beta>0$, and $\omega(b,d)=0$ otherwise. A natural choice to assess the cost of $(S_r,\mc{E}_r)$ is the MST of the associated graph $c(S_r) = \omega(\mc{E}^*_{S_r})$. 
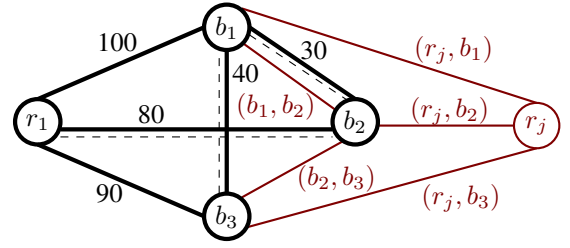
\begin{figure}[t!]
\centering
\begin{tikzpicture}
\pgfdeclarelayer{bg}    
\pgfsetlayers{bg,main}
\draw[ultra thick] (-1,4.75) circle (.3) node[anchor=center] {$b_3$};
\draw[ultra thick](-0.7,7.25)to (0.7,6.3);
\draw[dashed](-0.7,7.15)to (0.65,6.25);
\draw[ultra thick] (-1,7.25) circle (.3) node[anchor=center] {$b_1$};
\draw[ultra thick] (-3.5,6) circle (.3) node[anchor=center] {$r_1$};
\draw (-2.45,7.05) node[anchor=center] {100};
\draw (0.15,6.95) node[anchor=center] {30};
\draw (-2,6.1) node[anchor=center] {80};
\draw (-0.75,6.65) node[anchor=center] {40};
\draw (-2.55,5.05) node[anchor=center] {90};
\draw[ultra thick] (0.7,6) circle (.3) node[anchor=center] {$b_2$};
\draw[ultra thick](-1.3,7.25)to (-3.5,6.3);
\draw[ultra thick](-1,5.05)to (-1,6.95);
\draw[dashed](-1.1,5.05)to (-1.1,6.95);
\draw[ultra thick](-3.5,5.7)to (-1.3,4.75);
\draw[dashed](-3.2,5.8)to (0.4,5.8);
\draw[ultra thick](0.4,5.9)to (-3.2,5.9);
\begin{pgfonlayer}{bg}
  \draw[red!50!black] [thick ] (3.1,5.95) circle (.3) node[anchor=center] {$r_j$};
  \draw[red!50!black] [thick](-0.8,7.5)to (3.1,6.25);
  \draw[red!50!black] [thick](-0.8,5)to (0.55,5.75);
  \draw[red!50!black] [ thick](-0.8,7.05)to (0.45,6.15);
  \draw[red!50!black] [ thick](3.1,5.65)to (-0.75,4.6);
  \draw[red!50!black] [thick](1,5.95)to (2.8,5.95);
  \draw[red!50!black] (2,6.95) node[anchor=center] {$(r_j,b_1)$};
  \draw[red!50!black] (1.95,6.15) node[anchor=center] {$(r_j,b_2)$};
  \draw[red!50!black] (2.1,5) node[anchor=center] {$(r_j,b_3)$};
  \draw[red!50!black] (0.45,5.25) node[anchor=center] {$(b_2,b_3)$};
  \draw[red!50!black] (-0.35,6.2) node[anchor=center] {$(b_1,b_2)$};
\end{pgfonlayer}
\end{tikzpicture}
\caption{Example cost network from retailer $r_1$, and its MST (dashed). Other retailers $r_j$ also incur similar costs (in red).}
\label{p4fig:examplegraph}
\end{figure}
\begin{example}
  For $S_{r_1}=\{r_1,b_1,b_2,b_3\}$, with a cost network given in Fig.~\ref{p4fig:examplegraph}. This network consists of five different edges $\mc{E}_{r_1}=\{(r_1,b_1),(r_1,b_2),(r_1,b_3),(b_1,b_2),(b_1,b_3)\}$. There are seven possible trees containing all nodes. The MST contains the edges $\{(r_1,b_2),(b_1,b_2),(b_1,b_3)\} $ with cost $c(S_{r_1}) = 150$.
\end{example}
With the cost of a coalition, its value function can be obtained by turning the consumer's \emph{minimum cost spanning tree game}\cite{Granot1982} into a \emph{costs-saving} game. Now we can define the characteristic function of the cooperative game $(\mc{N},\nu)$ defined for all coalitions $S_r$ satisfying Assumption~\ref{assump:retailer_coalition} as:
\begin{equation}
\nu({S_r}) = \sum\limits_{b \in {\mc{B}_r}} \alpha_b\omega( {r},{b})  -  c (S_r),
\label{p4eq:coalvalue}
\end{equation}
where $\alpha_b \geq 1$ for all $b\in S_r$, and $\nu(\cdot)$ equals the sum of savings of all individual consumers, that is, the cost for being the only consumer in the coalition minus the cost of the coalition with other consumers. The following assumption is technical and will allow us to link the game $(\mc{N},\nu)$ with the optimisation problems~\eqref{p4eq:opt_mi}.
%
\begin{assumption}
  For each $b\in\mc{B}$, the function $U_b(\cdot,\cdot)$ in the profit $\Pi_b(P,\lambda) = U_b(P,r) - \lambda P$  can be written as $U_b(P,r) = \tilde{U}_b(P) - \omega(r,b)P^2$ with $\tilde{U}_b(P)$ satisfying Assumption~\ref{assum:costs}.
  \label{assum:cost_weights}
\end{assumption}

\subsection{Coalition Forming}
\label{p4sec:algo}
In this section, we construct the arrangement of retailers and buyers, as shown in Algorithm~\ref{alg:coal_form}, using an iterative approach that converges to a point on the \emph{Pareto front} of the optimisation problem~\eqref{p4eq:opt_mi}. The Pareto front is defined as
\begin{definition}
  Consider a function $F\colon \Rset^{|\mc{R}|}\times\Rset^{|\mc{B}|}\times \mc{M}_{\mc R,B} \to \Rset^{|\mc{R}| + |\mc{B}|}$, a partial order relation $\preceq$ on $\Rset^{|\mc{R}| + |\mc{B}|}$ and $\mc{Y} = F(\Rset^{|\mc{R}|},\Rset^{|\mc{B}|},\mc{M}_{\mc R,B})$. The \emph{Pareto Front} of $(F,\preceq)$ is 
  \begin{equation}
    \label{eq:pareto}
    \mc{P}(F,\preceq) = \{y\in\mc{Y}\colon\not\exists y'\in\mc{Y},~y\preceq y'\ts{ and }y\neq y'\}.
  \end{equation}
  A point $y\in\mc{P}(F,\preceq)$ is \emph{Pareto Optimal}.
  \label{def:pareto}
\end{definition}
This set contains all those points $(x,\mc{C})$ that cannot be dominated by any other point in the partial order $\preceq$. The concatenation of the objectives of \eqref{p4eq:opt_mi} results in a vector of profits from retailers and buyers $(\Pi_{r_1},\ldots,\Pi_{r_{|\mc{R}|}},\Pi_{b_1},\ldots,\Pi_{b_{|\mc{B}|}}) = F(x,\mc{C})\in\Rset^{|\mc{R}| + |\mc{B}|}$ where $x = (\lambda_\mc{R},P_\mc{B})$. Using a scalarisation approach, see \cite{Giagkiozis2015}, we can define the following problem for the network for any $\mu\in\Rset^{|\mc{R}|+|\mc{B}|}$ with $\md{1}^\top\mu = 1$ and $\mu \in [0,1]^{|\mc{R}|+|\mc{B}|}$ as
\begin{equation}
  \label{eq:scalarised_sol}
  \eta_{\mu}(\mc{C}) = \max_{x\in\mc{X}(\mc{C})}\mu^\top F(x,\mc{C}).
\end{equation}
A solution of \eqref{eq:scalarised_sol} $x^*_\mu \in \mc{X}(\mc{C})$ corresponds to a point on the Pareto front contained on the set $\{\mc{C}\}\times\mc{X}(\mc{C})$ as shown in \cite{Miettinen1999}. The function $\eta_\mu\colon\mc{M}_{\mc R,B}\to\Rset$ determines the best achievable scalarised welfare/utility when the market is structured following $\mc{C}\in\mc{M}_{\mc R,B}$. The desired object is $\max_{\mc{C}\in\mc{M}_{\mc R,B}}\eta_\mu(\mc{C})$ but its direct computation may require large discrete optimisations. In contrast, Algorithm \ref{alg:coal_form} navigates the space $\mc{M}_{\mc R,B}$ from an initial partition $\mc{C}^0$ towards a $\mc{C}^*$ based on a steepest ascent on the partition / coalitions. The scalarised optimisation problem~\eqref{eq:scalarised_sol} needs to be solved at each iteration, but this problem could be split into $|\mc R|$ distinct problems, each solved in a distributed manner (see Proposition~\ref{prop:separable}). Each consumer $b\in\mc{B}$ then changes retailer, leading to a new partition $\tilde{\mc{C}}$, that yields improvement in $\eta_\mu(\tilde{\mc C})$. The Algorithm then accepts the best-improving move if it exists; otherwise, it stops. When no single-buyer deviation improves $\eta_\mu$: local optimum for the improvement-path relation has been reached. In the next section, we show that Algorithm~\ref{alg:coal_form} stops after a finite number of iterations and the outcome $(x^*,\mc{C}^*)_\mu$ lies on the Pareto front for each $\mu$. 

\begin{algorithm}[t!]
  \label{alg:coal_form}
  \SetAlgoLined
  \KwData{$\forall r\in\mc{R}$, Initial guess $S_r$, $P_r^{max}$, and $C_r(\cdot,\cdot)$; $\forall b\in\mc{B}$, $U_b(\cdot,\cdot)$.}
  \KwResult{$\forall r\in\mc{R}$, optimal $S_r^*,~\lambda_r^*$; $\forall b\in\mc{B}$, optimal $P_b^*$.}
  \Repeat{$\{S_r\}_{r\in\mc{R}}$ is unchanging}{
    Compute $(\lambda_\mc{R}^\star,P_\mc{B}^\star)\in\arg\max\ \eta_\mu\big(\mc{C}\big)$\\
    Set $\eta_{\text{curr}} \leftarrow \eta_\mu\big(\mc{C}\big)$\\
    Initialise $\Delta_{\max}\leftarrow 0$ and $b^\star \leftarrow \emptyset,\ r^\star \leftarrow \emptyset$\\
    \ForEach{$b\in\mc{B}$}{
      Find $r_0\in\mc{R}$ such that $b\in S_{r_0}$\\
      \ForEach{$r\in\mc{R}\setminus\{r_0\}$}{
        $\tilde S_{r_0}\leftarrow S_{r_0}\setminus\{b\}$,\quad $\tilde S_r\leftarrow S_r\cup\{b\}$\\
        $\tilde S_{r'}\leftarrow S_{r'}$ for all $r'\in\mc{R}\setminus\{r_0,r\}$\\
        Compute $(\lambda_\mc{R}^\star,P_\mc{B}^\star)\in\arg\max\ \eta_\mu\big(\tilde{\mc{C}}\big)$\\
        Set $\eta_{\text{cand}} \leftarrow  \eta_\mu\big(\tilde{\mc{C}}\big)$, $\Delta \leftarrow \eta_{\text{cand}}-\eta_{\text{curr}}$\\
        \If{$\Delta>\Delta_{\max}$}{
          $\Delta_{\max}\leftarrow\Delta$, $b^\star\leftarrow b$,\quad $r^\star\leftarrow r$\\
        }
      }
    }
    \If{$\Delta_{\max}>\varepsilon$}{
      Let $r_0$ be the unique retailer such that $b^\star\in S_{r_0}$\;
      $S_{r_0}\leftarrow S_{r_0}\setminus\{b^\star\}$,\quad $S_{r^\star}\leftarrow S_{r^\star}\cup\{b^\star\}$\\
    }
  }
\caption{Coalition formation algorithm.}
\end{algorithm}
\subsection{Properties of the Coalitional Game $(\mc{N},\nu)$}
\label{p4sec:mcstvspc}
In this section, we show that coalitions formed using Algorithm~\ref{alg:coal_form} are stable in a game-theoretic sense. For $r\in\mc{R}$ and $S_r\subset\mc{N}$, the MST induces a partial order in the elements of $S_r$ with respect to $r$, \ie $b\succ_{r} d$ if $d$ lies in the path connecting $b$ to $r$.

We will first investigate the existence of a solution to the \emph{minimum cost spanning tree} (MCST) game, namely, the existence of a nonempty core given a retailer's coalition $S_r$. Non-emptiness of the core is a known property of convex games; moreover, \cite{Granot1982} has shown that MCST and convex games are PC games and all PC games possess a non-empty core. We are now ready to present the following results:
\begin{theorem}
  Suppose Assumption~\ref{assump:retailer_coalition} holds. For $r\in\mc{R}$, $( S_r,c)$ is a PC game with $c\colon 2^{S_r}\to\Rset$ the MST cost.
\label{thm:retailer_PC}
\end{theorem}
\begin{proof}
  For a consumer $b\in\mc{B}\cap S_r$, and $[b]\subset S_r$ an MST ordered set with $b\succ_r b_j$ where $b_j\in[b]$ satisfies $c([b]) > c([b_j])$ and $[b_j]\subset [b]$. Now, given a finite  $T\subset S_r\setminus [b]$, there exist a path $\phi(v_0,v_f)$ in $(S_r,\mc{E}_r)$ such that its end points $v_0,v_f\in S_r$ satisfy $v_0\in T$ and $v_f\in [b]$. The minimum spanning tree for $T\cup[b]$ contains all the edges from the path $\phi(r,b)$ and all those edges $\mc{E}_{b}^T = \bigcup_{t\in T}\bigcup_{b\in[b]}\phi(b,t)$, \ie $c(T\cup[b]) = \omega(\mc{E}_b^T) + c([b])$.  There are two possible cases: the set $T$ contains elements $d\succ_r b$ or $d\not\succ_r b$. For the former, the increment $c(T\cup [b]) - c([b]) = \omega(\mc{E}_b^T)$ which implies\[\begin{split}c(T\cup [b]) - c([b]) & = \omega(\mc{E}_b^T)+c([b_k]) - c([b_k])\\ c(T\cup [b]) - c([b]) & = c(T\cup[b_k]) - c([b_k])\end{split}\]for all $b_k\succ_r b$. For the second case, the set $T\subset S_r\setminus [b]$ contains at least a player that dominates $b\in\mc{B}\cap S_r$. The set of edges $\mc{E}_b^T$ contains two components: the path $\phi(d,b)$ and the remaining edges $\mc{E}_b^T\setminus\phi(d,b)$. Clearly, the increment \[c(T\cup [b]) - c([b]) = \omega(\phi(d,b))+ \omega(\mc{E}_b^T\setminus\phi(d,b)).\]On the other hand, for $b\succ_r b_k$,\[c(T\cup [b_k]) - c([b_k]) = \omega(\phi(d,b_k)) + \omega(\mc{E}_{b_k}^T\setminus\phi(d,b_k)).\]Since $b_k\succ_r d$, then $\omega(\phi(d,b_k)) >\omega(\phi(d,b))$ and $\omega(\mc{E}_{b_k}^T\setminus\phi(d,b_k)) = \omega(\mc{E}_b^T\setminus\phi(d,b))$. As a result,\[c(T\cup [b]) - c([b])<c(T\cup [b_k]) - c([b_k])\]for all $T\subset S_r\setminus[b]$. The MCST $(S_r,c)$ is a PC game.
\end{proof}
\begin{coro}
  Suppose Assumption~\ref{p4ass:union} holds. For $r\in\mc{R}$, The MCST $( S_r,c)$ has a non-empty core.
  \label{lem:core}
\end{coro}
The proof of Corollary~\ref{lem:core} follows from \cite[Theorem 1]{Granot1982}. The non-emptiness of the core for cost networks enables us to study and demonstrate how competition among retailers results in greater profits, which is reflected in the cooperative game $(\mc{N},\nu)$ being subadditive. In the next result, we show this is the case. 
\begin{theorem}
  Suppose Assumption~\ref{assump:retailer_coalition} holds. The coalitional game with multiple energy retailers $(\mc{N},\nu)$ with value function $\nu(\cdot)$ defined in~\eqref{p4eq:coalvalue} is subadditive, \ie for $r,s\in\mc{R}$
  \begin{equation}\label{p4eq:subadd}
    v(S_r \cup S_s)\leq v(S_s)+v(S_r).
  \end{equation}
\label{p4thm:comps}
\end{theorem}
\begin{proof}
  For $r,s\in\mc{R}$ and from Assumption~\ref{assump:retailer_coalition}, $\nu(S_r\cup S_s)=0$. On the other hand, $\nu(S_r) \geq 0$ since~\eqref{p4eq:coalvalue} depends on the sum of direct connection costs $\omega(r,b)$ with $b\in\mc{B}\cap S_r$, which is at worst equal to the MST cost $c(S_r)$. Therefore \eqref{p4eq:subadd} holds for any $r,s\in\mc R$.
\end{proof}
As mentioned in \cite{Chakraborty2019}, subadditivity alone is not sufficient to establish the stability of the game or the satisfaction of members of a given coalition, \ie there is no incentive to disband coalitions. To do this, we utilise the following result on concavity.
\begin{theorem}\label{p4thm:concavity}
  Suppose Assumption~\ref{assump:retailer_coalition} holds. The game with multiple energy retailers $(\mathcal{N},\nu)$ satisfies for any $r,s\in\mc{R}$,
  \begin{equation} \label{p4eq:concav}
    \nu(S_r \cup S_s)+\nu(S_r \cap S_s)\leq v(S_r)+v(S_s).
  \end{equation}
\end{theorem}
\begin{proof}
  As consequence of~\eqref{p4eq:opt_mi} and Assumption~\ref{assump:retailer_coalition}, $v(S_i \cup S_j)=0$ and $v(S_i \cap S_j)=v(\emptyset)=0$. From Theorem~\ref{p4thm:comps}, $v(S_r)\geq0$ holds for all $r\in\mc{R}$. Therefore \eqref{p4eq:concav} holds.
\end{proof}
A straight consequence of the above theorem is:
\begin{coro}[Coalition savings]
  For a retailer $r\in\mc{R}$, if its associated coalition $S_r$ consists of two or more consumers, then for all $b \in S_r$, $\nu(\{r,b\})\leq \nu(S_r)$.
  \label{coro:savings}  
\end{coro}
\begin{proof}
  From the definition of $\nu(\cdot)$, $\nu(\{r,b\}) = \omega(r,b) - c(\{r,b\}) = 0$. The result follows.
\end{proof}
\begin{remark}
In both Theorem~\ref{thm:retailer_PC} and \ref{p4thm:concavity} similar inequalities are used to denote either convexity and concavity respectively. As mentioned in \cite{Granot1982}, the difference lies in the nature of the characteristic function used in the game: for the former, $c(\cdot)$ is a cost function and measures how expensive it is for players to form a coalition, whereas for the latter, $\nu(\cdot)$ represents savings made by players joining a coalition.
\end{remark}
\subsection{The game $(\mc{N},\nu)$ and the optimisation~\eqref{p4eq:opt_mi}}
\label{sec:game-mcn-nu}

In this section, we formalise how the equilibrium notions of the cooperative game $\mc{N},\nu$ and the coalition formation process in Algorithm~\ref{alg:coal_form} are intertwined with the optimisation problems in~\eqref{p4eq:opt_mi}. We begin by asserting the concavity of linear combinations of 
\begin{lemma}
  Suppose Assumptions~\ref{assum:costs} and~\ref{assum:cost_weights} hold. Let\footnote{Given a function $y = F(x_1,x_2)$, we denote $\partial_{x_1}y$ as the partial derivative with respect to $x_1$, and $\partial_{x_1,x_2}^2y$ as the second derivative with respect to $x_1$ and $x_2$.}, for all $b\in\mc{B}\cap S_r$, $[U]_{bb} = \partial_{P_b,P_b}^2U_b(P_b)$, $[U]_{bd} = 0$, $[A]_{b,d} = \partial_{P_b,P_d}^2C_r$, $s_b = -1 + \partial_{\lambda_r,P_b}^2C_r$, and $a = \partial_{\lambda_r,\lambda_r}^2C_r$. Fix any $S_r\in\mc{C}$, and a set of scalars $(\mu_r,\mu_{\mc{B}\cap S_r}) \in (0,1]^{1+|S_r\cap\mc{B}|}$. Suppose $\exists \alpha>0$, satisfying $(\sqrt{\alpha} +\mu_r\sqrt{\bar{\mu}})^2< \mu_r$ where $\bar{\mu}$ is the maximum eigenvalue of $-A^{-1/2}U A^{-1/2}$, such that $-\mu_{\mc{B}\cap S_r}U^{-1}\mu_{\mc{B}\cap S_r} \leq \alpha a$, then $\xi(\lambda_r,P_{\mc{B}\cap S_r}) = \mu_r\Pi_r(\lambda_r,\sum_{b\in\mc{B}\cap S_r}P_b) + \sum_{b\in\mc{B}\cap S_r}\mu_{b}\Pi_b(P_b,\lambda_r)$ is concave.
  \label{lem:scalar_max}
\end{lemma}
\begin{proof}
  The concavity of the desired function can be determined by examining its second derivative. The gradients of $\xi$ are \[
    \begin{split}
      \partial_{\lambda_r}\xi =& \sum_{b\in S_r\cap\mc{B}}(\mu_r-\mu_b)P_b - \mu_r\partial_{\lambda_r}C_r\bigl(\lambda_r,P_{S_r}\bigr) \\
      \partial_{P_b}\xi =& (\mu_r-\mu_b)\lambda_r - \mu_r\partial_{P_b}C_r\bigl(\lambda_r,P_{S_r}\bigr)+ \mu_b\partial_{P_b}U_b\bigl(P_b\bigr).
    \end{split}
  \]
  The Hessian is
  \[
    \begin{split}
      \partial_{\lambda_r,\lambda_r}^2\xi = & - \mu_r\partial^2_{\lambda_r,\lambda_r}C_r\bigl(\lambda_r,P_{S_r}\bigr)\\
      \partial_{\lambda_r,P_b}^2\xi = & \mu_r-\mu_b - \mu_r \partial_{\lambda_r,P_b}^2C_r\bigl(\lambda_r,P_{S_r}\bigr)\\
      \partial_{P_b,P_b}^2\xi = & - \mu_r\partial_{P_b,P_b}^2C_r\bigl(\lambda_r,P_{S_r}\bigr)+ \mu_b\partial_{P_b,P_b}^2U_b\bigl(P_b\bigr)\\
      \partial_{P_b,P_d}^2\xi =  & - \mu_r\partial_{P_b,P_d}^2C_r\bigl(\lambda_r,P_{S_r}\bigr)
    \end{split}
  \]
  The Hessian can be written as \[
    H = \begin{bmatrix}
      -a & (\mu_{\mc{B}\cap S_r} - \mu_r s)^\top \\
      \mu_{\mc{B}\cap S_r} - \mu_r s & U - \mu_rA\\
    \end{bmatrix}
  \]
  where $A\in\Rset^{|S_r\cap\mc{B}||\times |S_r\cap\mc{B}|}$ is, by Assumption~\ref{assum:costs}, is positive definite and $U\in\Rset^{|S_r\cap\mc{B}||\times |S_r\cap\mc{B}|}$ is diagonal and negative definite.   The negative definiteness of $H$ can be obtained using the Schur complement $(\mu_{\mc{B}\cap S_r}-\mu_rs)^\top (\mu_r A  - U)^{-1} (\mu_{\mc{B}\cap S_r}-\mu_rs) \leq \mu_r a$.  Using matrix inequalities, we obtain $\mu_r A - U \succeq - U \iff(\mu_r A - U)^{-1} \preceq -U^{-1}$. Consequently,
\[-(\mu_{\mc{B}\cap S_r}-\mu_rs)^\top U^{-1} (\mu_{\mc{B}\cap S_r}-\mu_rs) \leq \mu_r a.\]
By Young's inequality, for $\varepsilon > 0$,
\[\begin{split}
    -(\mu_{\mc{B}\cap S_r}-\mu_rs)^\top & U^{-1} (\mu_{\mc{B}\cap S_r}-\mu_rs) \leq \\
    -(1+\varepsilon) \mu_{\mc{B}\cap S_r}^\top &U^{-1} \mu_{\mc{B}\cap S_r} - \bigl(1 + \frac{1}{\varepsilon}\bigr)\mu_r^2 s^\top U^{-1} s \leq \\
    (1+\varepsilon)\alpha a + & \bigl(1 + \frac{1}{\varepsilon}\bigr)\mu_r^2 \beta a \leq \\
     ((1+\varepsilon)\alpha a + &\bigl(1 + \frac{1}{\varepsilon}\bigr)\mu_r^2 \bar{\mu} )a.  
  \end{split}\]
In the last inequalities, we have used our hypothesis $ -\mu_{\mc{B}\cap S_r}^\top U^{-1} \mu_{\mc{B}\cap S_r} \leq \alpha a$ and imposed $-s^\top U^{-1} s \leq \bar{\mu} a$. Consider now $\varepsilon = \mu_r\sqrt{\frac{\bar{u}}{\alpha}}$, we obtain
\[
  \begin{split}
    -(\mu_{\mc{B}\cap S_r}-\mu_rs)^\top & U^{-1} (\mu_{\mc{B}\cap S_r}-\mu_rs) \leq\\
    & (\sqrt{\alpha} +\mu_r\sqrt{\bar{\mu}})^2a \leq \mu_r a.
  \end{split}
\]
The last inequality was our hypothesis. What is left to prove is $-s^\top U^{-1} s \leq \bar{\mu} a$. To this end, we use the S-procedure to show $-U\preceq \hat{\beta} A$. Notice that $s^\top A^{-1} s \leq a$ implies $-s^\top U^{-1} s \leq \hat{\beta} a$ provided $\hat{\beta} \geq \bar{\mu}$. 

\end{proof}
Lemma~\ref{lem:scalar_max} provides guidelines for choosing the scalarisation weights to maintain a maximisation problem. We use this result to prove the following:
\begin{proposition}[Separability of the value function]
  Suppose Assumptions~\ref{assump:retailer_coalition}--\ref{assum:cost_weights} hold. Fix a weight vector $\mu\in\Rset^{|\mathcal R|+|\mathcal B|}_+$ such that $\md{1}^\top\mu=1$. Suppose, in addition, that for each $r\in\mc{R}$ and each feasible coalition $S_r=\{r\}\cup B_r$ with $B_r\subseteq \mc{B}$, there is a nonempty feasible set $\mc{X}_r(S_r)$ such that for every $\mc{C}\in\mc{M}_{\mc R,B}$, \(\mc{X}(\mc{C}) \;=\; \prod_{r\in\mc{R}} \mc{X}_r\bigl(S_r\bigr)\). Then, for all $S_r\in\mc{C}$, there exists $V_{r,\mu}\colon \mc{X}_r(S_r)\times\mc{C}\to\Rset$ such that the scalarised optimisation problem \eqref{eq:scalarised_sol} with $x_r = (\lambda_r,P_{S_r\cap\mc{B}})$ can be written as
  \begin{equation}
    \eta_\mu(\mc{C}) = \sum_{r\in\mc{R}} V_{r,\mu}(x_r,S_r).
    \label{eq:sop_sep}
  \end{equation}
\label{prop:separable}
\end{proposition}
\begin{proof}
  Each partition $\mc{C}$ determines disjoint sets $S_r$ for each $r\in\mc{R}$. By construction, the optimisation problems~\eqref{p4eq:optprice_mi} for a fixed $\mc{C}$ and for each $r\in\mc{R}$ depend only on $\lambda_r$ and $S_r\cap\mc{B}$ this implies that the overall constraint set $\mc{X}(\mc{C}) = \Pi_{r\in\mc{R}}\mc{X}_r$ with $\mc{X}_r  = \{(\lambda_r,P_{S_r})\colon \sum_{b\in S_r}P_b \leq P_r^\ts{max}\}$. Similarly, the problems~\eqref{p4eq:optcons_mi} depend on $P_b$ and $\lambda_r$ for all $b\in\mc{B}\cap S_r$. Consider the scalarised problem \eqref{eq:scalarised_sol} for a collection of weights $\mu\in\Rset^{|\mc{R}|+|\mc{R}|}$,\[\eta_\mu(\mc{C}) = \max_{x\in\mc{X}(\mc{C})}\sum_{r\in\mc{R}}\mu_r\Pi_r(\lambda_r,P_{\mc{B}\cap S_r}) + \sum_{b\in\mc{B}}\mu_b\Pi_b(P_b,\lambda_r).\]By the previous argument, we can write the following\[
    \begin{split}
      \eta_\mu(\mc{C}) = \sum_{r\in\mc{R}} & \max_{(\lambda_r,P_{\mc{B}\cap S_r})\in\mc{X}_r} \mu_r\Pi_r(\lambda_r,P_{\mc{B}\cap S_r}) +  \\
                                           &  \sum_{b\in\mc{B}\cap S_r}\mu_b\Pi_b(P_b,\lambda_r)= \sum_{r\in\mc{R}} V_{r,\mu} (x_r,S_r).
    \end{split}\]
  where $x_r = (\lambda_r,P_{\mc{B}\cap S_r})$ and the maximisation is a well defined object following Lemma~\ref{lem:scalar_max}. 
\end{proof}

The following result is one of our main contributions: the finiteness of Algorithm~\ref{alg:coal_form}. The proof relies on showing the convergence point forms a \emph{Nash Equilibrium} of the game $(\mc{N},(\lambda_\mc{R},P_\mc{B}),(\Pi_\mc{R},\Pi_\mc{B}))$.
\begin{theorem}[Convergence of Algorithm \ref{alg:coal_form}] Suppose Assumptions \ref{assump:retailer_coalition}--\ref{assum:cost_weights} hold. For a given initial partition $\mc{C}_0\in\mc{M}_{\mc R,B}$, the sequence of partitions generated from applying Algorithm \ref{alg:coal_form} $,\mf{C} = \{\mc{C}_0,\mc{C}_1,\ldots,\mc{C}_n,\ldots\}$ satisfies , $\forall \mc{C}_0\in\Pi_{\mc{R,\mc{B}}}$, $\exists n_0>0$, such that  $\forall n\geq n_0$, implies $\mc{C}_{n+1} = \mc{C}_n$.
\label{thm:alg_convergence}
\end{theorem}
\begin{proof}
  The strategy for this proof follows the construction of a potential game in which players are $b\in\mc{B}$ and actions are the choice of retailers $a_b\in\mc{R}$. An action profile determines a partition $\mc{C}$ and hence the coalitions $S_r$. From \cite{Monderer1996}, a potential game requires a potential function and a notion of utility. For the latter, consider $u_b(a) = V_{r,\mu}(S_r) - V_{r,\mu}(S_r\setminus \{b\})$ as the utility of buyer $b$. The potential is the function $\eta_\mu(\mc{C})$ which by Proposition~\ref{prop:separable} is separable. The potential is $\Phi(a) = \sum_{r\in\mc{R}}\nu(S_r) + V_{r,\mu}(x_r,S_r)$ and it is easy to check $\Phi(a') - \Phi(a) = u_b(a) - u_b(a')$ for any profiles $a,a'$, \ie different choices of retailer assignment. By \cite{Monderer1996}, the game has a \emph{Nash Equilibrium}. Furthermore, we have $\eta_\mu(\mc{C}_{k+1}) > \eta_\mu(\mc{C}_{k})$; therefore, no cycles are admitted in the sequence of partitions. Hence, at most $|\mc{R}|^{|\mc{B}|} - 1$ improvements can occur which implies termination in finite time. 
\end{proof}
The following is a direct consequence of Theorem~\ref{thm:alg_convergence}
\begin{coro}
  Suppose Assumptions \ref{assump:retailer_coalition}--\ref{assum:cost_weights} hold and $\mu\in\Rset^{|\mc{R}|+|\mc{B}|}$ satisfy Corollary~\ref{lem:scalar_max}. Let $\mc{C}^*\in\mc{M}_{\mc R,B}$ termination point of Algorithm~\ref{alg:coal_form} and $x^* = (\lambda_\mc{R}^*,P_\mc{B}^*) = \arg\max \{\mu^\top F(x,\mc{C}^*)\colon x\in\mc{X}(\mc{C}^*)\}$. The pair $(x^*,\mc{C}^*)$ is Pareto Optimal for \eqref{p4eq:opt_mi} with the partial order induced by $\preceq_{\eta,\mu}$, \ie $(x,\mc{C})\preceq_{\eta,\mu}(y,\mc{D})$ if $\eta_\mu(\mc{C})\leq\eta_\mu(\mc{D})$ for all $x,y,\mc{C},\mc{D}$. 
\end{coro}
\begin{proof}
  This follows from the final partition being a Nash equilibrium of the potential game defined in the proof of Theorem~\ref{thm:alg_convergence}.
\end{proof}
A consequence of Theorem~\ref{p4thm:comps} and Assumption~\ref{assump:retailer_coalition} is that $\nu(\mc{N}) = 0$ since the grand coalition contains all retailers. This fact renders conventional cooperative game methods, see \cite{Bausoa,Chakraborty2019}, unusable. We now turn our attention to show that the partition $\mc{C}^*\in\mc{M}_{\mc R,B}$ is in fact a stable solution of the cooperative game $(\mc{N},\nu)$. The assessment of stability relies on the $\mbb{D}-$stability, introduced in \cite{Apt2006,Apt2009}, denominated defection defined as:
\begin{definition}
  Let $(\mc{N},\nu)$ be a cooperative game. A partition $\mc{C} = \{C_1,\ldots,C_L\}$ of $\mc{N}$ is $\mbb{D}-$stable if $\sum_{l=1}^L\nu(C_l) \geq \sum_{k=1}^K\nu(D_k)$ for any other partition $\mc{D} = \{D_1,\ldots,D_K\}$.
\label{def:Dstab}
\end{definition}
We note that Definition~\ref{def:Dstab} introduces a notion of social welfare. A $\mbb{D}-$stable partition is the partition no element would want to deviate from. The following result is a Corollary to Theorem~\ref{thm:alg_convergence} and shows that the resulting partition is stable.
\begin{coro}
  Suppose Assumptions~\ref{assump:retailer_coalition}--\ref{assum:cost_weights} hold. The termination partition $\mc{C}^*\in\mc{M}_{\mc R,B}$ from Algortithm~\ref{alg:coal_form} is $\mbb{D}-$stable.
  \label{cor:Dstab}
\end{coro}
\begin{proof}
  Let $\mc{C}^*$ be the terminal partition of Algorithm~\ref{alg:coal_form}. Consider the modified function $\hat{V}_{\mu} (S_r)= \nu(S_r) + V_{r,\mu}(x_r,S_r)$; the domain of $\hat{V}_{\mu}$ coincides with that of $\nu$. We have that the pair $(\mc{N},\hat{V}_{r,\mu})$ is a cooperative game. Following from Theorem~\ref{p4thm:comps} and Theorem~\ref{thm:alg_convergence}, the game is PC and the partition $\mc{C}^*$ yields strictly better results than any other $\mc{D}\in\mc{M}_{\mc R,B}$. Therefore, $\mc{C}^*$ is $\mbb{D}-$stable.
\end{proof}
\section{Risk Sharing and Reduction of Statistical Dispersion}
\label{p4sec:stats}
In this section, we provide a brief insight into the statistical implications for consumers when considering their demands as random variables. Motivated by \cite{Baeyens2013}, we consider the statistical properties of consumer demand to show that coalition formation entails lower risks for both consumers and retailers. These risks include lower aggregate demand or capacity violations for retailer $r\in\mc{R}$.  Consumers can lower this risk by acting together to join a retailer coalition, which allows them to increase collective profits and share the risk.

In this setting, we consider the power demands $P_b(t)\in[\underline{P}_b,\bar{P}_b]$, following the constraints in~\eqref{p4eq:optcons_mi}, as a stochastic process with a joint compactly supported. The power consumption of a coalition $S_r$ for $r\in\mc{R}$ is also a stochastic process $P_{S_r}(t)=\sum_{b\in S_r\cap\mc B}{P_{b}(t)}$ with compact support on $\mc{P}_{S_r} = [ 0, \sum_{b\in S_r\cap\mc B}\bar{P}_b]$, with  cummulative distribution function $\Phi_{S_r}(Q,t)=\mathbb{P}[P_{S_r}(t)\leq Q]$ for any $Q\in\Rset$. Considering a uniformly distributed random variable $U\sim Unif[t_0,t_f]$ such that $\mb{P}_{S_r} = P_{S_r}(U)$, the time averaged CDF of $P_{S_r}$ is
\begin{equation}
  F_{S_r}(P) = \mbb{P}[\mb{P}_{S_r}\leq P] = \frac{1}{T}\int_{t_0}^{t_f}\Phi_{S_r}(P,t)dt.
  \label{eq:cdf_Sr}
\end{equation}
%
The total consumer profit for a given price $\lambda_r$ is $\Pi_{S_r}(\mb{P}_{S_r},\lambda_r)=\sum_{b \in S_r\cap\mc B}{\Pi_{b}(P_{b}(U),\lambda_r)}$. The expected consumer profit is also a stochastic process \(  J_{S_r}(\lambda_r)=\mathbb{E}[\Pi_{S_r}({P}_{S_r},\lambda_r)] = \int_{\mc{P}_{S_r}} \Pi_{S_r}(P,\lambda_r) dF_{S_r}(P)\). Our first result states that the average profit is larger for $S_r$ than in smaller coalitions $\{r,b\}$.
\begin {proposition}
  Suppose Assumptions \ref{assump:retailer_coalition} and \ref{assum:costs} hold. Fix $\mc{C}\in\mc{M}_{\mc R,B}$, for $r\in\mc R$, the following holds \emph{almost surely}, 
\begin{equation}
  J_{S_r}(\lambda_r)\geq\sum_{b \in S_r\cap \mc B}{J_{\{r,b\}}(\lambda_r)}.
  \label{eq:av_prof_ae}
\end{equation}
\label{prop:av_prof_ae}
\end{proposition}
\begin{proof}
By the properties of the sum of random variables \cite{Lemons2002} and of homogeneity and superadditivity of functions of stochastic processes \cite{Baeyens2013}; condition~\eqref{eq:av_prof_ae} is fulfilled as a direct consequence of Corollary~\ref{coro:savings}, where the value of a coalition is always greater than or equal to the value of a single consumer with a retailer.  
\end{proof}
The above result establishes that coalitions always bring larger collective profits for consumers. These benefits can be directly attributed to the attenuation of statistical dispersion from aggregation; this phenomenon has also been explored in \cite{Rajagopal}. We now turn our attention to defining the issue of risk within a coalition $S_r$: capacity overloads, cost volatility, and curtailment. In this paper, we focus on capacity overloads. Consider $Y_r = \max(\mb{P}_{S_r} - P_r^\ts{max},0)$ and the associated \emph{Value at Risk} for $\alpha\in(0,1)$ as $\mathrm{VaR}_\alpha (Y_r) = \inf\{y\in\Rset\colon \mbb{P}[Y_r\leq y] \geq \alpha\}$, the \emph{Conditional Value at Risk} from \cite{Rockafellar2005} is
\begin{equation}
  \label{eq:CVaR}
  \mathrm{CVaR}_\alpha(Y_r) = \mbb{E}[Y_r | Y_r\geq \mathrm{VaR}_\alpha(Y_r))].
\end{equation}

The following results show that coalition formation minimises risk and statistical dispersion.
\begin{theorem}
  Suppose Assumptions~\ref{assump:retailer_coalition} and \ref{assum:costs} hold. Suppose, in addition, For each $b\in\mc B$, $\widetilde P_b\in[0,\overline P_b]$ almost surely and $\{\mb{P}_b\}_{b\in\mc B}$ are independent. Let $P_r^{\max}>0$ and  $\alpha\in(0,1)$ fixed. The risk functional is \(  \Psi_\alpha(\mc{C}) := \sum_{r\in\mc R} \mathrm{CVaR}_\alpha\big(Y_r\big)\). Consider the following \emph{risk-improvement (steepest-descent) dynamics}:
  starting from any $\mc{C}_0$, at iteration $k$ choose
  \begin{equation}
    \label{eq:risk_dyn}
  \mc{C}_{k+1} =
  \begin{cases}
    \arg\min_{\mc{C}'\in \mc{N}(\mc{C}_k)} \Psi_\alpha(\mc{C}') &  \Psi_\alpha(\mc{C}_{k+1})<\Psi_\alpha(\mc{C}_k)\\
    \mc{C}_k & \ts{otherwise},
  \end{cases}    
  \end{equation}
where $\mc N(\mc{C})$ is the set of assignments obtained from $\mc{C}$ by a \emph{single-buyer reassignment} from Algorithm~\ref{alg:coal_form}. Then, $i)$ there exists $N_R>0$ such that $\mc{C}_{k+1} = \mc{C}_k = \hat{\mc{C}}$ for all $k>N_R$; $ii)$ $\hat{\mc{C}} = \{\hat{S}_{r_1},\ldots,\hat{S}_{r_{|\mc{R}|}}\}$ is \emph{one-buyer risk-stable}: for every buyer $b\in \hat{S}_r\cap \mc B$ and every retailer $r'\in\mc R\setminus\{r\}$, 
\[
  \Psi_\alpha(\hat{\mc{C}})\ \le\ \Psi_\alpha\big(\hat{\mc{C}}^{\,b\to r'}\big),
\]
where $\hat{\mc{C}}^{\,b\to r'}$ is the partition obtained from $\hat{\mc{C}}$ by moving $b$ to $r'$, \ie $S_{r'}^{\,b\to r'} = \hat{S}_{r'} \cup \{b\}$. 
  \label{thm:risk_minimisation}
\end{theorem}

\begin{proof}
  Each $\mb{P}_b$ is bounded almost surely, then $\mb{P}_{S_r}$ is bounded almost surely for all $r\in\mc{R}$, and $Y_r$ is bounded and nonnegative, so $\mbb{E}[(Y_r-\eta)_+]<\infty$ for all $\eta\in\Rset$ and $\mathrm{CVaR}_\alpha(Y_r)$ is finite following Rockafellar's formula \cite{Rockafellar2005}. Consequently $\Psi_\alpha(\mc{C})$ is well-defined and finite for all $\mc{C}\in\mc{M}_{\mc R,B}$. For $i)$, by Proposition~\ref{prop:counting}, $|\mc{M}_{\mc R,B}| = |\mc{R}|^{|\mc{B}|}$ and $\Psi_\alpha(\mc{C}_k)$ is nondecreasing, as a result no periodic orbits exist in the recursion~\eqref{eq:risk_dyn}. Hence only finitely many updates can occur, so the dynamics terminate
after finitely many iterations at some $\hat{\mc{C}}$.

For $ii)$, the dynamics~\eqref{eq:risk_dyn} stop at $\hat{\mc{C}}$ when there is no neighbour $\mc{C}'\in\mc N(\hat{\mc{C}})$ such that $\Psi_\alpha(\mc{C}')<\Psi_\alpha(\hat{\mc{C}})$. By definition of $\mc N(\hat{\mc{C}})$, every neighbour corresponds to moving exactly one buyer $b\in\mc B$ from its current retailer $\hat{r}$ to some other retailer $r'\in\mc R\setminus\{\hat{r}\}$. Therefore, for all such $(b,r')$, \(\Psi_\alpha(\hat{\mc{C}}) \le \Psi_\alpha\big(\hat{\mc{C}}^{\,b\to r'}\big)\), the statement is proved.
\end{proof}

Theorem~\ref{thm:risk_minimisation} gives a precise sense in which coalition formation can be viewed as a \emph{risk-reduction} process: the monotone partition dynamics terminate at a risk-stable partition $\hat{\mc{C}}$. No buyer can unilaterally move to another retailer and further reduce the system’s $\mathrm{CVaR}$ of overload. In particular, the terminal coalition structure is an equilibrium with respect to downside overload risk: it is locally optimal for tail risk and directly links coalition formation to the mitigation of capacity-violation events. We note that the construction in this case is similar to the one presented for the deterministic case and could be easily integrated into the framework of Algorithm~\ref{alg:coal_form}.

\section{Examples}
\label{sec:examples}
 To illustrate our scheme, we have formulated two scenarios. The first one consists of a \ac{uG} with a single retailer, whereas the second one considers two additional retailers; both scenarios have five consumers and three retailers. The parameters for both are listed in Table \ref{p4tab:param} and the network topology is shown in Figure~\ref{p4fig:CostNetworks}.
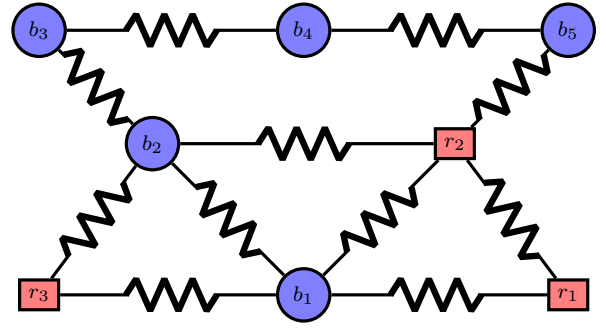
\begin{figure}
  \centering
  \begin{tikzpicture}
    \node[draw, very thick, fill=blue!50, circle, minimum size=0.6cm] (L1) at (-2.5,-3) {\footnotesize $b_{1}$};
    \node[draw, very thick, fill=blue!50, circle, minimum size=0.6cm] (L2) at (-4.5,-1) {\footnotesize $b_{2}$};
    \node[draw, very thick, fill=blue!50, circle, minimum size=0.6cm] (L3) at (-6,0.5) {\footnotesize $b_{3}$};
    \node[draw, very thick, fill=blue!50, circle, minimum size=0.6cm] (L4) at (-2.5,0.5) {\footnotesize $b_{4}$};
    \node[draw, very thick, fill=blue!50, circle, minimum size=0.6cm] (L5) at (1,0.5) {\footnotesize $b_{5}$};
    \node[draw,very thick,fill=red!50,rectangle, minimum size=0.1cm] (G1) at (1,-3) {\footnotesize $r_1$};
    \node[draw,very thick,fill=red!50,rectangle, minimum size=0.1cm] (G2) at (-0.5,-1)  {\footnotesize $r_2$};
    \node[draw,very thick,fill=red!50,rectangle, minimum size=0.1cm] (G3) at (-6,-3) {\footnotesize $r_3$};

    \draw[very thick] (L4) to[resistor] (L5);
    \draw[very thick] (L4) to[resistor] (L3);
    \draw[very thick] (L2) to[resistor] (L3);
    \draw[very thick] (L2) to[resistor] (L1);
    \draw[very thick] (L2) to[resistor] (G3);
    \draw[very thick] (L2) to[resistor] (G2);
    \draw[very thick] (L5) to[resistor] (G2);
    \draw[very thick] (G2) to[resistor] (L1);
    \draw[very thick] (G2) to[resistor] (G1);    
    \draw[very thick] (L1) to[resistor] (G3);
    \draw[very thick] (L1) to[resistor] (G1);      
    
  \end{tikzpicture}
  \caption{Network topology. The set of buyers $\mc{B} = \{b_1,b_2,b_3,b_4,b_5\}$ and retailers $\mc{R} = \{r_1,r_2,r_3\}$}
  \label{p4fig:CostNetworks}
\end{figure}
We propose for each $r\in\mc{R}$ and $b\in\mc{B}$ the following cost and utility functions
\begin{subequations}
  \begin{align}
    C_r(\lambda,P)&={\alpha_{{r}}}\lambda^2 P^2+\nu(S_r),\label{p4eq:cost}\\
    U_b(P,r)&={ \alpha_{{1,b}}}P^{\alpha_{2,b}}-\omega(r,b)P^2+\kappa_{{r}}\delta_{{b}}^{{r}}\label{p4eq:util}
  \end{align}
  \label{eq:cost_util}
\end{subequations}
where $\alpha_r>0$, $\alpha_{1,b}>0$, $\alpha_{2,b}\in (0,1)$. For $C_r(\cdot)$, the first term is related to the generation costs for each retailer, and the second to the savings from consumers connected to that retailer. The utility for a buyer includes a concave function of the power consumed, a base payment, $-\omega(r,b)$, for joining $S_r$, and a and a \emph{subsidy} ${\kappa_{{r}}\delta_{{b}}^{{r}}}\geq 0$ to consumer $b$ from retailer $r$. The subsidy term is equivalent to the Banzhaf power index employed in cooperative games, which quantifies how pivotal a player is within a coalition; see \cite{Algaba2019}. A form of the subsidy term is also used in \cite{Namerikawa} as an incentive tool.

\begin{table}[b!]
  \caption{Parameters for Retailers and Consumers.\label{p4tab:param}}
  \centering
  \begin{tabular}{cccccc}
    \hline 
\textbf{Retailer} & $ \alpha_{{r}_i}$ & $\kappa_{{r}_i}$ & $ \underline\lambda_{i}$ & $\overline \lambda_i $& $P^\ts{max}_{i}$\\ \hline
$r_1$      & 1e-1       & 0.65    & 0.01     & 2   & 1 \\
$r_2$      & 7e-1       & 0.64    & 0.01     & 2   & 1 \\
$r_3$      & 5e-1       & 0.63    & 0.01     & 2   & 1 \\ \hline
&&&&\\
\end{tabular}
\begin{tabular}{ccccc}
\hline
\textbf{Consumer} & $ \alpha_{1,{b}_j}$ & $\alpha_{2,{b}_j}$& $\underline \zeta_{{b}_j}$ & $\overline \zeta_{{b}_j}$  \\ \hline
$b_1$   & 1   & 0.42 & 0  & 1  \\
$b_2$   & 1   & 0.42 & 0  & 1  \\
$b_3$   & 1   & 0.32 & 0  & 1  \\
$b_4$   & 1   & 0.46 & 0  & 1  \\
$b_5$   & 1   & 0.26 & 0  & 1  \\ \hline
\end{tabular}
\end{table}
The properties of Algorithm~\ref{alg:coal_form}  are illustrated in Figures~\ref{fig:coal_conv} and \ref{fig:vf}. For the former, Figure~\ref{fig:coal_conv} shows how, for a choice of $\mu_1$, $\mu_2$, and $\mu_3$ satisfying the conditions stated in Lemma~\ref{lem:scalar_max}, the algorithm converges in finite time despite being initialised at all the possible partitions in $\mc{M}_{\mc R,B}$. The algorithm converges towards a set of partitions for all possible initialisations. This suggests Algorithm~\ref{alg:coal_form} admits a notion of invariance, \ie there are partitions in $\mc{M}_{\mc R,B}$ that are yield the maximum value for $\eta_\mu$. For the different choices of scalarising weights, the solution has the form $\mc{R}\times \{B_1,\ldots,B_{|\mc{R}|}\}$ containing all possible combinations of retailers and subsets of buyers. For \[\mu_1 = (0.004,0.004,  0.094, 0.004,  0.094, 0.207, 0.387,  0.207),\] the partitions where the algorithm converges are \(\mc{R} \times \{\{b_2\},\{b_4\},\{ b_1,b_3,b_5\}\}\) and \(\mc{R} \times \{\{b_2,b_5\},\{b_4\},\{ b_1,b_3\}\}\). Similarly, \[ \mu_2 = 0.004,0.094,0.004,0.049,0.049,0.207,0.207,0.387),\] yields \(\mc{R} \times  \{\{b_2,b_5\},\{b_3\},\{ b_1,b_4\}\}\), \(\mc{R} \times \{\{b_1,b_5\},\{b_3\},\{ b_2,b_4\}\}\), and \(\mc{R} \times \{\{b_3,b_5\},\{b_1\},\{ b_2,b_4\}\}\). Lastly, for \[ \mu_3 = (0.049, 0.004,  0.094,  0.004, 0.049,0.387,0.207,0.207),\]Algorithm~\ref{alg:coal_form} terminates in  \(\mc{R} \times \{\{b_1,b_3,b_5\},\{b_4\},\{ b_2\}\}\). In all these cases, as shown in Figure~\ref{fig:vf}, the function $\eta_\mu(\cdot)$ increases monotonically along the partition path, as guaranteed by Theorem~\ref{thm:alg_convergence}. The distribution of convergence steps is shown in Figure~\ref{fig:kdis}; for the given weight vectors, the average number of convergence steps is between $3$ and $4$. 
\begin{figure}[t!]
  \centering
 \input{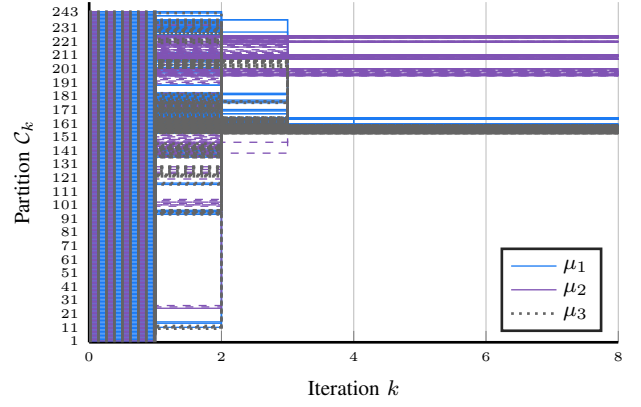}
  \caption{Convergence of Algorithm~\ref{alg:coal_form} starting at a partition $\mc{C}_0\in\mc{M}_{\mc R,B}$ using different scalarisation constants $\mu_1,\mu_2,\mu_3 \in [0,1]^{|\mc{R}|+|\mc{B}|}$ (\eqref{fig:eta_mu1}, \eqref{fig:eta_mu2}, and \eqref{fig:eta_mu3} respectively).}
  \label{fig:coal_conv}
\end{figure}
\begin{figure}[t!]
  \centering
\input{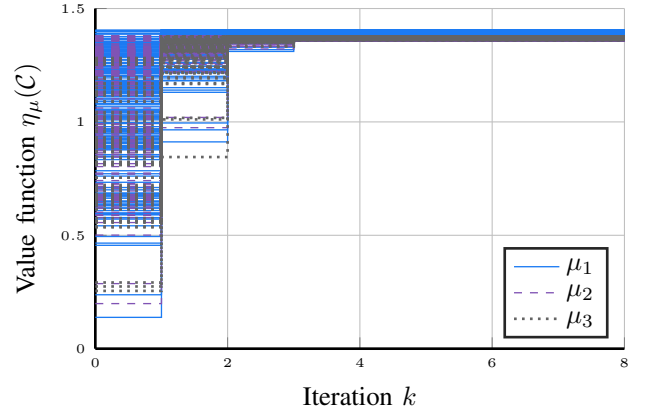}
  \caption{Monotonic evolution of the function $\eta_\mu(\cdot)$ along the iterations of Algorithm~\ref{alg:coal_form} for all possible starting partitions $\mc{C}\in\mc{M}_{\mc R,B}$. Each line \eqref{fig:eta_mu1}, \eqref{fig:eta_mu2}, and \eqref{fig:eta_mu3} corresponds to three different scalarisation vectors $\mu_1$, $\mu_2$, and $\mu_3$ respectively.}
  \label{fig:vf}
\end{figure}
\begin{figure}[t!]
  \centering
%
\definecolor{mycolor1}{rgb}{0.10980,0.48235,0.95294}%
\definecolor{mycolor2}{rgb}{0.50196,0.32549,0.70588}%
\definecolor{mycolor3}{rgb}{0.39216,0.39216,0.39216}%
\begin{tikzpicture}

\begin{axis}[%
width=7cm,
height=4.5cm,
at={(0cm,0cm)},
scale only axis,
area legend,
xmin=0.0800000000000006,
xmax=6.02,
xlabel={$K^{max}_\mu$},
xmajorgrids,
ymin=0,
ymax=0.6,
ymajorgrids,
axis x line*=bottom,
axis y line*=left,
legend style={legend cell align=left,align=left,draw=white!15!black}
]
\addplot[ybar,bar width=0.5cm,bar shift=-0.2cm,draw=black,fill=mycolor1] plot table[row sep=crcr] {%
1	0.0246913580246914\\
2	0.238683127572016\\
3	0.497942386831276\\
4	0.230452674897119\\
5	0.00823045267489712\\
};\addlegendentry{$\mu_1$}\label{fig:his_mu1}
\addplot[ybar,bar width=0.5cm,bar shift=0cm,draw=black,fill=mycolor2] plot table[row sep=crcr] {%
1	0.0493827160493827\\
2	0.366255144032922\\
3	0\\
4	0.452674897119342\\
5	0.131687242798354\\
};\addlegendentry{$\mu_2$}\label{fig:his_mu2}
\addplot[ybar,bar width=0.5cm,bar shift=0.2cm,draw=black,fill=mycolor3] plot table[row sep=crcr] {%
1	0.0246913580246914\\
2	0.222222222222222\\
3	0\\
4	0.576131687242798\\
5	0.176954732510288\\
};\addlegendentry{$\mu_3$}\label{fig:his_mu3}
\end{axis}
\end{tikzpicture}
  \caption{Histogram of the distribution of the convergence steps of Algorithm~\ref{alg:coal_form} for all possible intial partions and three scalarising vectors $\mu_1$, $\mu_2$, and $\mu_3$ (\eqref{fig:his_mu1}, \eqref{fig:his_mu2}, \eqref{fig:his_mu3} respectively).}
  \label{fig:kdis}
\end{figure}
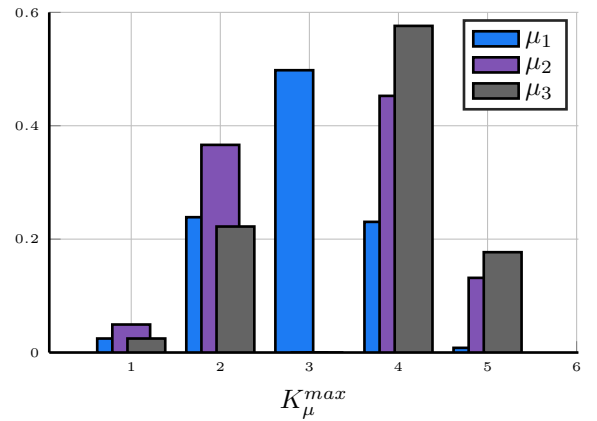

In Figure~\ref{fig:coal_inv_heat}, we have initialised the algorithm at each $\mc{C}\in\mc{M}_{\mc R,B}$ and have used a dense sample of the set of admissible scalarising weights consisted of $1000$ points. The convergence of the algorithm shows clearly that there exists an \emph{equilibrium set} $\mf{C}\subset\mc{M}_{\mc R,B}$ that is attractive for Algorithm~\ref{alg:coal_form}. In the figure, the vertical axis corresponds to the initial partition, and the horizontal one represents the terminal one; the colour of each entry $(\mc{C}_f,\mc{C}_0)$ determines the number of times $\mc{C}_0$ has ended in $\mc{C}_f$ using different scalarisation vectors. In Figures~\ref{fig:par_coor_equilibrium} and \ref{fig:par_coor_var}, we explore the performance attained by Algorithm~\ref{alg:coal_form} on the set $\mf{C}\subset\mc{M}_{\mc R,B}$ for both objectives and variables for the scalarised optimisation problems~\eqref{p4eq:opt_mi}. Each polyline corresponds to one terminal coalition structure $\mc C(\mu)\in\mf{C}$ (for some sampled $\mu$), and the value on each vertical axis gives the resulting objective (buyer profits $\Pi_{b_1},\dots,\Pi_{b_5}$ and retailer profits $\Pi_{r_1},\Pi_{r_2},\Pi_{r_3}$). Since every $\mc C(\mu)$ is obtained as a global optimum of a weighted-sum scalarisation, the plotted outcomes are supported weak Pareto optima. Crossings between axes highlight the trade-offs among objectives along the supported Pareto set.
\begin{figure*}[t!]
  \centering
%
\input{figs/colormap_def.tex} 
\begin{tikzpicture}

\begin{axis}[%
width=16cm,
height=4.5cm,
at={(0cm,0cm)},
scale only axis,
axis on top,
xmin=0.5,
xmax=243.5,
xtick={118,120,122,154,156,158,160,162,164,196,198,200,202,204,206,208,210,212,220,222,224},
xticklabel style={rotate=90},
xlabel={Partitions in  $\mc{M}_{\mc{R},\mc{B}}$},
xmajorgrids,
ymin=0.5,
ymax=243.5,
ytick={1,11,21,31,41,51,61,71,81,91,101,111,121,131,141,151,161,171,181,191,201,211,221,231,241},
ylabel={Partitions in  $\mc{M}_{\mc{R},\mc{B}}$},
ymajorgrids,
major grid style={
  line width=0.2pt,
  draw=gray!20,                   
},
minor grid style={
  line width=0.1pt,
  draw=gray!20,                   
},
colormap name=mycolormap,
colorbar,
colorbar style={
  width=0.2cm,
  ytick={0, 0.2, 0.4, 0.6, 0.8, 1},
  yticklabels = {$0$, $0.1$, $0.2$, $0.3$, $0.4$, $0.48$},
  yticklabel style={/pgf/number format/fixed,
    /pgf/number format/precision=2},
    },
legend style={legend cell align=left,align=left,draw=white!15!black}
]
        \addplot graphics[
            xmin=0, xmax=243,
            ymin=0, ymax=243,
        ] {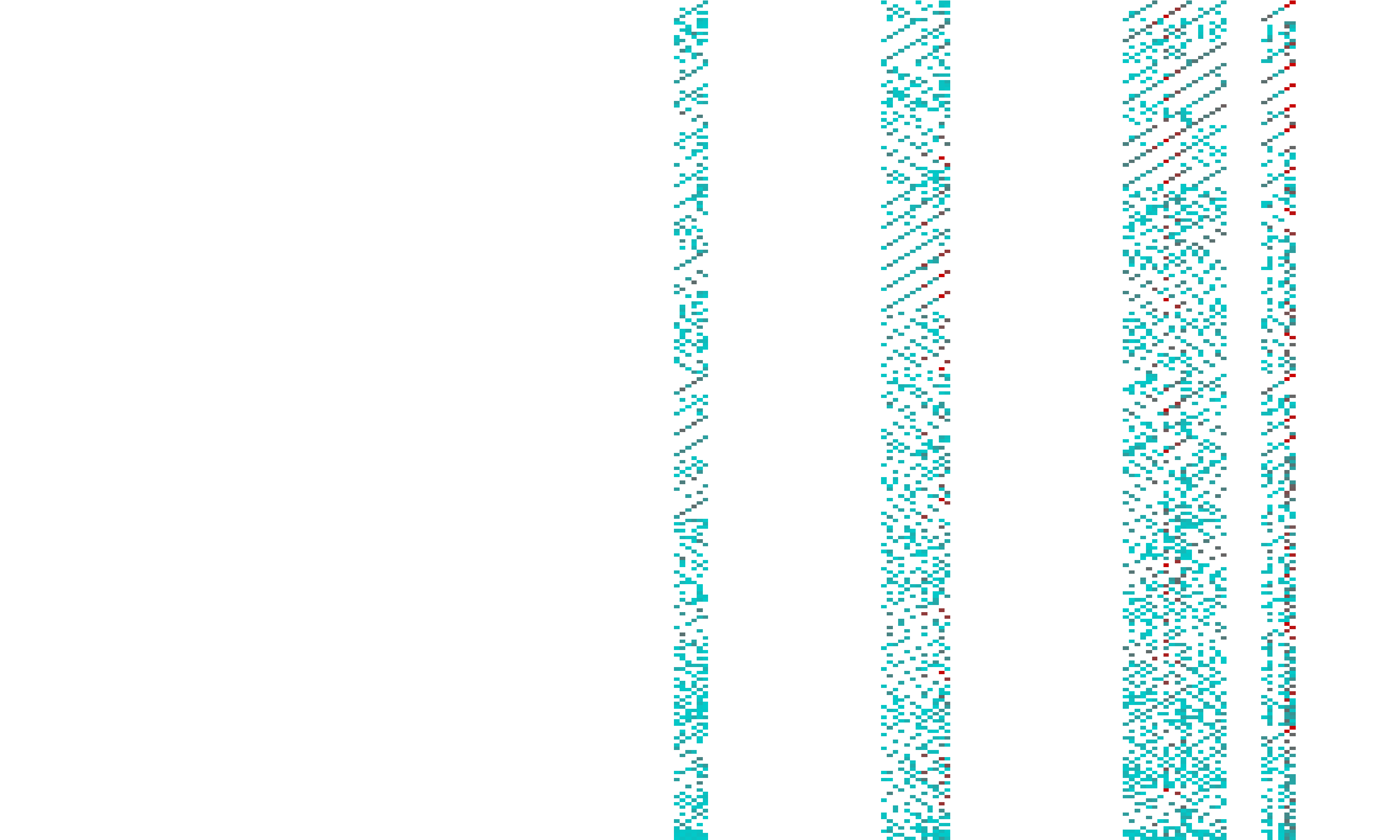};
\end{axis}
\end{tikzpicture}%

  \caption{Normalised transition count matrix induced by Algorithm~\ref{alg:coal_form} over the complete partition set $\mc{M}_{\mc R,B}$. Rows index the initial partition $\mc C_0\in\mc{M}_{\mc R,B}$ and columns index the terminal partition $\mc C^\star(\mu,\mc C_0)$ returned by the algorithm. For a finite sample $\mc U\subset[0,1]^{|\mc R|+|\mc B|}$ of scalarisation vectors satisfying Lemma~\ref{lem:scalar_max}, the entry $C(i,k)$ reports the fraction of $\mu\in\mc U$ for which the algorithm initialised at partition $i$ terminates at partition $k$ (darker colours indicate higher frequency).}
  \label{fig:coal_inv_heat}
\end{figure*}
\begin{figure}[t!]
  \centering
  \input{figs/colormap_def.tex} 
\begin{tikzpicture}

\begin{axis}[%
width=7cm,
height=4.5cm,
at={(0cm,0cm)},
scale only axis,
axis on top,
xmin=0.5,
xmax=18,
xtick={2,4,6,8,10,12,14,16},
xticklabels = {$\Pi_{b_1}$,$\Pi_{b_2}$,$\Pi_{b_3}$,$\Pi_{b_4}$,$\Pi_{b_5}$,$\Pi_{r_1}$,$\Pi_{r_2}$,$\Pi_{r_3}$},
xmajorgrids,
ymajorgrids,
ymin=-0.6,
ymax=0.6,
ytick={-0.55,0.55},
ylabel={\footnotesize Objective value},
ylabel shift=-0.7cm,
major grid style={
  line width=0.2pt,
  draw=gray!20,                   
}
]
        \addplot graphics[
            xmin=0.95, xmax=17,
            ymin=-0.65, ymax=0.6,
        ] {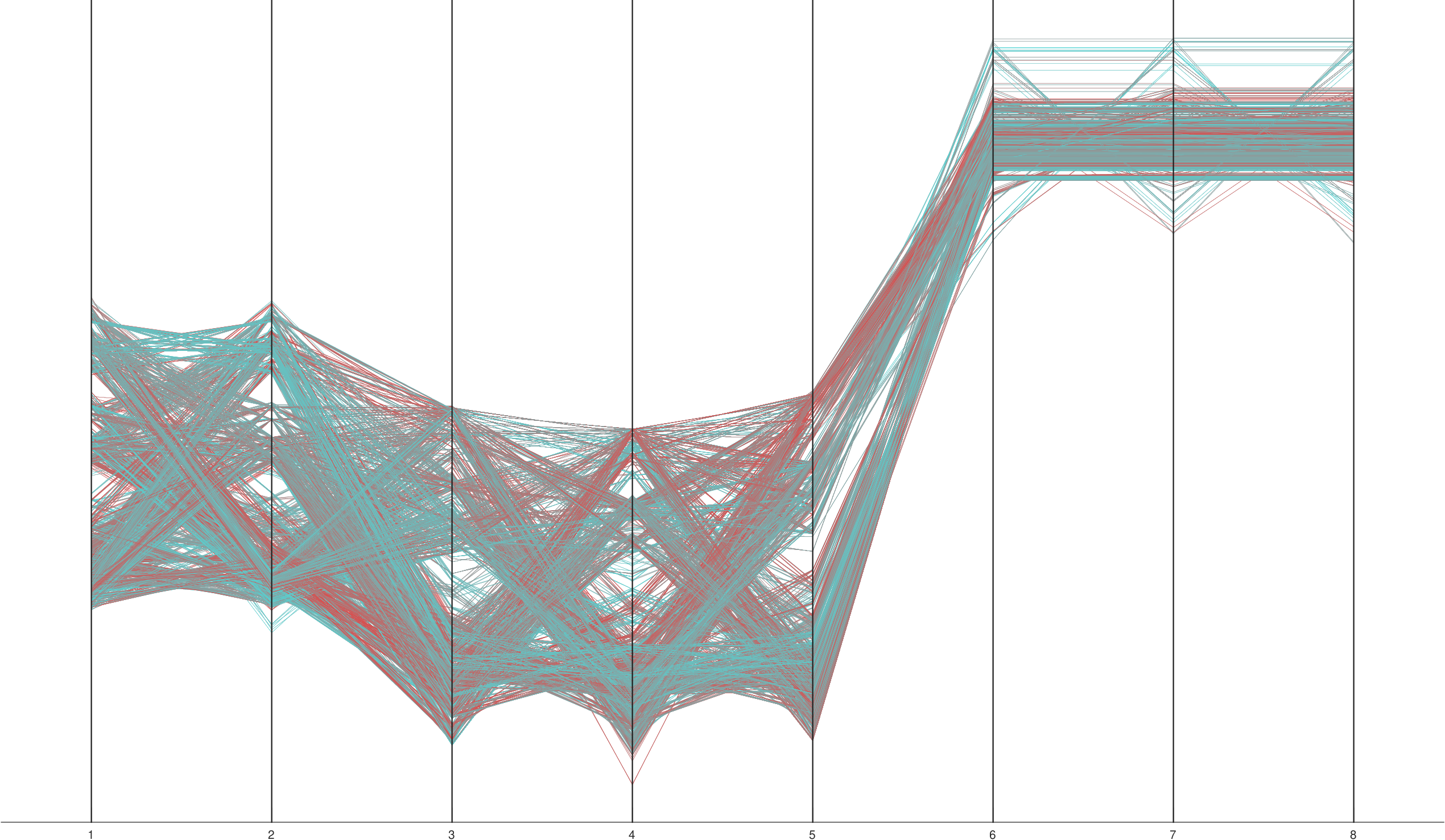};
\end{axis}
\end{tikzpicture}%

  \caption{Parallel coordinates plot of the objective vectors obtained from the equilibrium partition set $\mf{C}\subset\mc{M}_{\mc R,B}$ for~\eqref{p4eq:opt_mi}. Each polyline corresponds to a terminal solution $\mc C(\mu)\in\mf{C}$ (for some sampled $\mu$) and reports the buyer profits $\Pi_{b_i}$ and retailer profits $\Pi_{r_j}$. Since $\mc C(\mu)$ solves a weighted-sum scalarisation globally, all plotted outcomes are supported weak Pareto optima; crossings indicate trade-offs.}

  \label{fig:par_coor_equilibrium}
\end{figure}
\begin{figure}[t!]
  \centering
  \input{figs/colormap_def.tex} 
\begin{tikzpicture}

\begin{axis}[%
width=7cm,
height=4.5cm,
at={(0cm,0cm)},
scale only axis,
axis on top,
xmin=0.5,
xmax=18,
xtick={2,4,6,8,10,12,14,16},
xticklabels = {$P_{b_1}$,$P_{b_2}$,$P_{b_3}$,$P_{b_4}$,$P_{b_5}$,$\lambda_{r_1}$,$\lambda_{r_2}$,$\lambda_{r_3}$},
xmajorgrids,
ymajorgrids,
ymin=0,
ymax=2.05,
ytick={0.02,2},
ylabel={\footnotesize variable value},
ylabel shift=-0.7cm,
major grid style={
  line width=0.2pt,
  draw=gray!20,                   
}
]
        \addplot graphics[
            xmin=0.95, xmax=17,
            ymin=-0.115, ymax=2.1,
        ] {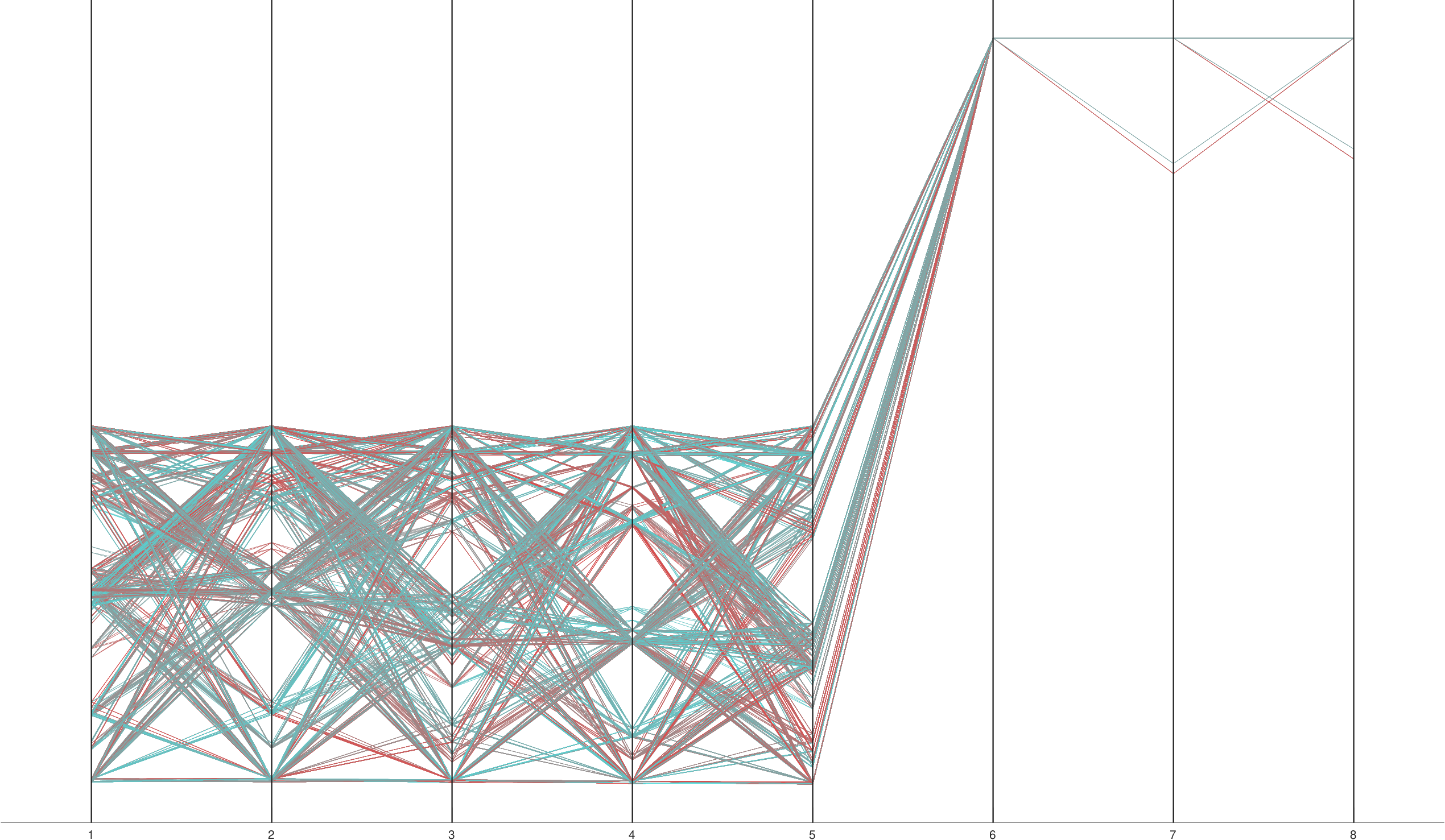};
\end{axis}
\end{tikzpicture}%
  \caption{Parallel coordinates plot of the vector of decision variables obtained from the equilibrium partion set $\mf{C}\subset\mc{M}_{\mc R,B}$. The algorithm returns prices $\lambda_r$ that are almost constant at the highest value allowed by the constraints. The trade-offs correspond to consumers' selection of consumed power, $P_b$.}
  \label{fig:par_coor_var}
\end{figure}
\section{Conclusions and future work}
\label{sec:concl-future-work}
In this paper, we have proposed a novel demand-side management scheme for a network with multiple retailers. We have modelled competition amongst retailers as a coalitional game leading to a collection of  $|\mc{N}|$ coupled mixed-integer optimisation problems in which prices, power demands, and the network partition are decision variables. Our approach consisted of proposing a coalition-forming algorithm based on multi-objective optimisation principles that aim to maximise retailers' profit and consumers' welfare. We have shown that our algorithm exhibits finite convergence and recovers a subset of weak Pareto-optimal trade-offs. Furthermore, the outcome of our proposed algorithm is stable in a game-theoretic sense.  Numerical results from an academic case study illustrated the method's behaviour and showed that the resulting equilibrium partition set contains multiple meaningful trade-offs between the competing objectives, in line with standard properties of Pareto fronts. Finally, we extended the framework to a risk-sharing setting by replacing the expected-value objective with a conditional value-at-risk criterion. Future work will relax the assumption that each node acts exclusively as either a consumer or a load and will study how contract duration influences network performance, coalition stability, and price recomputation.
\bibliography{bibliography}
\bibliographystyle{IEEEtran}
\end{document}